%% file: ms.tex
\documentclass[11pt]{article}

\usepackage[utf8]{inputenc}

\usepackage{amsmath,amsfonts,amsthm,amssymb,color}
\usepackage[usenames,dvipsnames,svgnames,table]{xcolor}
\definecolor{darkgreen}{rgb}{0.0,0,0.9}
\usepackage{mathtools}
\usepackage{authblk}
\usepackage{fullpage}
\usepackage{parskip}
\usepackage{comment}
\usepackage{tikz}
\usepackage{bbm}
\usepackage{dsfont}
\usepackage[sc]{mathpazo}
\usepackage[basic]{complexity}
\usepackage{algorithm2e}
\usepackage[colorlinks=true,
citecolor=OliveGreen,linkcolor=BrickRed,urlcolor=BrickRed,pdfstartview=FitH]{hyperref}
\usepackage[capitalize,nameinlink]{cleveref}
\usepackage{tcolorbox}
\usepackage[short]{optidef}

\newtcolorbox{wbox}
{
	colback  = white,
}

\SetKwInOut{Input}{Input}
\SetKwInOut{Output}{Output}
\SetKwFunction{Uncross}{\textsc{Uncross}}
\SetKwFunction{MergeUncross}{\textsc{MergeUncross}}
\SetKwFunction{PerfectMatching}{\textsc{PerfectMatching}}
\SetKwBlock{InParallel}{in parallel do}{end}
\SetKwFor{ParallelFor}{for}{in parallel do}{end}

\let\R\relax
\newcommand*{\R}{\mathbb{R}}

\newcommand*{\suppress}[1]{}

\newcommand*{\cR}{\mathcal{R}}

\makeatletter
\def\thm@space@setup{%
	\thm@preskip= 10pt
	\thm@postskip=\thm@preskip 
}
\makeatother

\makeatletter
\renewcommand{\paragraph}{%
	\@startsection{paragraph}{4}%
	{\z@}{5pt}{-1em}%
	{\normalfont\normalsize\bfseries}%
}
\makeatother

\newtheorem{theorem}{Theorem}
\newtheorem{lemma}{Lemma}

\newtheorem{corollary}{Corollary}

\theoremstyle{definition}
\newtheorem{definition}{Definition}

\newtheorem{remark}{Remark}

\newtheorem{example}{Example}

\newenvironment{fminipage}%
{\begin{Sbox}\begin{minipage}}%
		{\end{minipage}\end{Sbox}\fbox{\TheSbox}}

\newcommand\QQ{\boldsymbol{\mathit{Q}}}

\newcommand\ZZ{\boldsymbol{\mathit{Z}}}

\newcommand*{\ov}[1]{\bar{#1}}
\newcommand*{\un}[1]{\underline{#1}}

\title{New Characterizations of Core Imputations of\\
 Matching and $b$-Matching Games} 

\author[1]{Vijay V.~Vazirani\footnote{Supported in part by NSF grants CCF-1815901 and CCF-2230414.}}

\affil[1]{University of California, Irvine}

\date{}

\begin{document}
	\maketitle
	
	\begin{abstract}
We give new characterizations of core imputations for the following games:

\begin{enumerate}
	\item The assignment game.
	\item Concurrent games, i.e., general graph matching games having non-empty core.
	\item The unconstrained bipartite $b$-matching game (edges can be matched multiple times).
	\item The constrained bipartite $b$-matching game (edges can be matched at most once).
\end{enumerate}

\bigskip

The classic paper of Shapley and Shubik \cite{Shapley1971assignment} showed that core imputations of the assignment game are precisely optimal solutions to the dual of the LP-relaxation of the game. Building on this, Deng et al. \cite{Deng1999algorithms} gave a general framework which yields analogous characterizations for several fundamental combinatorial games. Interestingly enough, their framework does not apply to the last two games stated above. In turn, we show that some of the core imputations of these games correspond to optimal dual solutions and others do not.  This leads to the tantalizing question of understanding the origins of the latter.  

We also present new characterizations of the profits accrued by agents and teams in core imputations of the first two games. Our  characterization for the first game is stronger than that for the second; the underlying reason is that the characterization of vertices of the Birkhoff polytope is stronger than that of the Balinski polytope.
	\end{abstract}

\bigskip
\bigskip
\bigskip
\bigskip
\bigskip
\bigskip
\bigskip
\bigskip
\bigskip
\bigskip
\bigskip
\bigskip
\bigskip
\bigskip
\bigskip
\bigskip
\bigskip
\bigskip

\pagebreak

\input{Intro}

\input{Related}

\input{Vertices}

\input{Edges}

\input{Negotiate}

\input{General}

\input{Vertices2}

\input{b-Matching}

\input{b-Constrained}

\input{b-General}

\input{discussion}

\input{ack}

	\bibliographystyle{alpha}
	\bibliography{refs}

\end{document}

%% file: Intro.tex
\section{Introduction}
\label{sec.intro}

The matching game forms one of the cornerstones of cooperative game theory and the {\em core} is a  quintessential solution concept in this theory; the latter captures all possible ways of distributing the total worth of a game among individual agents in such a way that the grand coalition remains intact, i.e., a sub-coalition will not be able to generate more profits by itself and therefore has no incentive to secede from the grand coalition. The matching game can also be viewed as a matching market in which utilities of the agents are stated in monetary terms and side payments are allowed, i.e., it is a {\em transferable utility (TU) market}. For an extensive coverage of these notions, see the book by Moulin \cite{Moulin2014cooperative}.    

The classic paper of Shapley and Shubik \cite{Shapley1971assignment} showed that the set of core imputations of the assignment game as the set of optimal solutions to the dual of the LP-relaxation of the maximum weight matching problem in the underlying graph. Among their other insights was a characterization of the two ``antipodal'' points --- imputations which maximally favor one side of the bipartition\footnote{Much like the top and bottom elements in a lattice of stable matchings.} --- in the core of this game. This in-depth understanding makes the assignment game a paradigmatic setting for studying the core; in turn, insights gained provide valuable guidance on profit-sharing in real-life situations. 

Deng et al. \cite{Deng1999algorithms} distilled the ideas underlying the Shapley-Shubik Theorem to obtain a general framework (see Section \ref{sec.Deng}) which helps characterize the core of several games that are based on fundamental combinatorial optimization problems, including maximum flow in unit capacity networks both directed and undirected, maximum number of edge-disjoint $s$-$t$ paths, maximum number of vertex-disjoint $s$-$t$ paths, maximum number of disjoint arborescences rooted at a vertex $r$, and concurrent games (defined below). 

In this paper, we study the core of the assignment game and some of its generalizations, including  two versions of the bipartite $b$-matching game (Section \ref{sec.b-matching-game}); in the first version (Section \ref{sec.b-uncon-core}) edges can be matched multiple number of times and in the second, edges can be matched at most once (Section \ref{sec.b-con-core}). The intriguing aspect of the latter two games is that they don't fall in framework of Deng et al.; see Section \ref{sec.Deng} for the reason. In turn, we show that some of the core imputations of these games correspond to optimal dual solutions and some not. This leads to a tantalizing question: is there a ``mathematical structure'' that produces the latter?

For the assignment game (Section \ref{sec.Complementarity}), we start with the realization is that despite the in-depth work of Shapely and Shubik, and the passage of half a century, there are still basic questions about the core  which have remained unexplored: 

\begin{enumerate}
	\item Do core imputations spread the profit more-or-less evenly or do they restrict them to certain well-chosen agents? If the latter, what characterizes these ``chosen'' agents?
	\item  By definition, under any core imputation, the sum of profits of two agents $i$ and $j$ is at least the profit they make by being matched, say $w_{ij}$. What characterizes pairs $(i, j)$ for which this sum strictly exceed $w_{ij}$?  
	\item How do core imputations behave in the presence of degeneracy?
	\end{enumerate}
	
An assignment game is said to be {\em degenerate} if the optimal assignment is not unique. Although Shapley and Shubik had mentioned this phenomenon, they brushed it away, claiming that ``in the most common case'' the optimal assignment will be unique, and if not, their suggestion  was to perturb the edge weights to make the optimal assignment unique. However, this is far from satisfactory, since perturbing the weights destroys crucial information contained in the original instance and the outcome becomes a function of the vagaries of the randomness imposed on the instance. 

The following broad idea helps answer all three questions. A well-known theorem in matching theory says that the LP-relaxation of the optimal assignment problem always has an integral optimal solution \cite{LP.book}. Therefore, the worth of the assignment game is given by the optimal objective function value of this LP. Next, the Shapley-Shubik Theorem says that the set of core imputations of this game are precisely the optimal solutions to the dual of this LP. These two facts naturally raise the question of viewing core imputations through the lens of complementarity; in turn, it leads to a resolution of all three questions.

The following setting, taken from \cite{Eriksson2001stable} and \cite{Biro2012computing}, vividly captures the issues underlying profit-sharing in an assignment game. Suppose a coed tennis club has sets $U$ and $V$ of women and men players, respectively, who can participate in an upcoming mixed doubles tournament. Assume $|U| = m$ and $|V| = n$, where $m, n$ are arbitrary. Let $G = (U, V, E)$ be a bipartite graph whose vertices are the women and men players and an edge $(i, j)$ represents the fact that agents $i \in U$ and $j \in V$ are eligible to participate as a mixed doubles team in the tournament. Let $w$ be an edge-weight function for $G$, where $w_{i j} > 0$ represents the expected earnings if $i$ and $j$ do participate as a team in the tournament. The total worth of the game is the weight of a maximum weight matching in $G$.

Assume that the club picks such a matching for the tournament. The question is how to distribute the total profit among the agents --- strong players, weak players and unmatched players --- so that no subset of players feel they will be better off seceding and forming their own tennis club. We will use this setting to discuss the issues involved in the questions raised above.

Under core imputations, the profit allocated to an agent is a function of the value he/she brings to the various sub-coalitions he/she belongs to, i.e., it is consistent with his/her negotiating power. Indeed, it is well known that core imputations provide profound insights into the negotiating power of individuals and sub-coalitions; Section \ref{sec.negotiation} illustrates this via some well-chosen examples. The first question provides further insights into this issue. Our answer to this question is that the core rewards only {\em essential} agents, namely those who are matched by {\em every} maximum weight matching, see Theorem \ref{thm.vertices}. 

Our answer to the second question is quite counter-intuitive: we show that a pair of players $(i, j)$ get overpaid by core allocations if and only if they are so incompetent, as a team, that they don't participate in any maximum weight matching! Since $i$ and $j$ are incompetent as a team,   $w_{ij}$ is small. On the other hand, a least one of $i$ and $j$ does team up with other agents in every maximum weight matching -- if not, $(i, j)$ would have been matched. Therefore, the sum of the profits of $i$ and $j$ exceeds $w_{ij}$ in at least one core imputation; this is shown in Theorem \ref{thm.edges}. 

Our insight into degeneracy is that it treats teams and agents in totally different ways, see  Section \ref{sec.degeneracy}. Section \ref{sec.related} discusses past approaches to degeneracy. 

Whereas the core of the assignment game is always non-empty, that of the general graph matching game can be empty. Deng et al. \cite{Deng1999algorithms} showed that the core of this game is non-empty if and only if the weights of maximum weight integral and fractional matchings concur, if so, we call them {\em concurrent games}; their core imputations are also precisely the set of optimal solutions to the dual LP. 

Next, we study the three questions, raised above, for concurrent games, Section \ref{sec.general}.  The answers obtained for the first two questions are weaker than those for the assignment game. The underlying reason is that the characterization of the vertices of the Balinski polytope, Theorem \ref{thm.Balinski}, for concurrent games, is weaker than the characterization of the vertices of the Birkhoff polytope, Theorem \ref{thm.Birkhoff}, for assignment games; in the former, vertices are half-integral matchings and in the latter, they are integral matchings. The answer to third question is identical to that of the assignment game.


%% file: Related.tex
\section{Related Works}
\label{sec.related}

An imputation in the core has to ensure that {\em each} of the exponentially many sub-coalitions is ``happy'' --- clearly, that is a lot of constraints. As a result, the core is non-empty only for  a handful of games, some of which are mentioned in the Introduction. A different kind of game, in which preferences are cardinal, is based on the stable matching problem defined by Gale and Shapley \cite{GaleS}. The only coalitions that matter in this game are ones formed by one agent from each side of the bipartition. A stable matching ensures that no such coalition has the incentive to secede and the set of such matchings constitute the core of this game.   

To deal with games having an empty core, e.g., the general graph matching game, the following two notions have been given in the past. The first is that of {\em least core}, defined by Mascher et al. \cite{Leastcore-Maschler1979geometric}. If the core is empty, there will necessarily be sets $S \subseteq V$ such that $v(S) < p(S)$ for any imputation $v$. The least core maximizes the minimum of $v(S) - p(S)$ over all sets $S \subseteq V$, subject to $v(\emptyset) = 0$ and $v(V) = p(V)$. This involves solving an LP with exponentially many constraints, though, if a separation oracle can be implemented in polynomial time, then the ellipsoid algorithm will accomplish this in polynomial time \cite{GLS}; see below for a resolution for the case of the matching game.  

A more well known notion is that of {\em nucleolus} which is contained in the least core.  After maximizing the minimum of $v(S) - p(S)$ over all sets $S \subseteq V$, it does the same for all remaining sets and so on. A formal definition is given below.

\begin{definition}
	\label{def.nucleolus}
	For an imputation $v: {V} \rightarrow \cR_+$, let $\theta(v)$ be the vector obtained by sorting the $2^{|V|} - 2$ values $v(S) - p(S)$ for each $\emptyset \subset S \subset V$ in non-decreasing order. Then the unique imputation, $v$, that lexicographically maximizes $\theta(v)$ is called the {\em nucleolus} and is denoted $\nu(G)$. 
\end{definition}

The nucleolus was defined in 1969 by Schmeidler \cite{Schmeidler1969nucleolus}, though its history can be traced back to the Babylonian Talmud \cite{Talmud-Aumann1985game}. It has several modern-day applications, e.g., \cite{Branzei2005strongly}. In 1998, \cite{Faigle1998nucleon} stated the problem of computing the nucleolus of the matching game in polynomial time. For the assignment game with unit weight edges, this was done in \cite{Nucleolus-Assign-Solymosi1994algorithm}; however, since the assignment game has a non-empty core, this result was of little value. For the general graph matching game with unit weight edges, this was done by Kern and Paulusma \cite{Kern2003matching}. Finally, the general problem was resolved by Konemann et al. \cite{Konemann2020computing}. However, their algorithm makes extensive use of the ellipsoid algorithm and is therefore neither efficient nor does it give deep insights into the underlying combinatorial structure. They leave the open problem of finding a combinatorial polynomial time algorithm. We note that the difference $v(S) - p(S)$ appearing in the least core and nucleolus has not been upper-bounded for any standard family of games, including the general graph matching game. 

A different notion was recently proposed in \cite{Va.general}, namely {\em approximate core}. That  paper gives an imputation in the $2/3$-approximate core for the general graph matching game, i.e., the total profit allocated to a sub-coalition is at least $2/3$ factor of the profit which it can generate by seceding. Moreover, this imputation can be computed in polynomial time and the bound is best possible, since it is the integrality gap of the natural underlying LP. This work used methodology developed in field of approximation algorithms, e.g., see \cite{ApproximationAlgs}, which uses multiplicative approximation as a norm, and yields polynomial time algorithms. 

Building on the work of \cite{Va.general}, Xiao et al. gave an efficient algorithm for obtaining a 2/3-approximate core allocation for $b$-matching games in general graphs \cite{b-matching-approximate}. We note that this game was called {\em multiple partners matching game} by Sotomayor \cite{Sotomayor1992multiple} and the bipartite $b$-matching game was called the {\em multiple partners assignment game}; however that paper did not study their core. 

Konemann et al. \cite{b-matching-Konemann} showed that computing the nucleolus of the constrained bipartite $b$-matching game is NP-hard even for the case $b=3$ for all vertices. Biro et al. \cite{Biro2012computing} showed that the core non-emptiness and core membership problems for the multiple partners matching game are solvable in polynomial time if $b \leq 2$ and are co-NP-hard even for $b = 3$.


Over the years, researchers have approached the phenomenon of degeneracy in the assignment game from directions that are different from ours. Nunez and Rafels \cite{Nunez-Dimension}, studied relationships between degeneracy and the dimension of the core. They defined an agent to be {\em active} if her profit is not constant across the various imputations in the core, and non-active otherwise. Clearly, this notion has much to do with the dimension of the core, e.g., it is easy to see that if all agents are non-active, the core must be zero-dimensional. They prove that if all agents are active, then the core is full dimensional if and only if the game is non-degenerate. Furthermore, if there are exactly two optimal matchings, then the core can have any dimension between 1 and $m-1$, where $m$ is the smaller of $|U|$ and $|V|$; clearly, $m$ is an upper bound on the dimension.

In another work, Chambers and Echenique \cite{Chambers2015core} study the following question: Given the entire set of optimal matchings of a game on $m = |U|$, $n = |V|$ agents, is there an $m \times n$ surplus matrix which has this set of optimal matchings. They give necessary and sufficient conditions for the existence of such a matrix.

%% file: Vertices.tex
\section{The Core of the Assignment Game}
\label{sec.Complementarity}

In this section, we provide answers to the three questions, for assignment games, which were 
 raised in the Introduction.

\input{Definitions}

\subsection{The first question: Allocations made to agents by core imputations}
\label{sec.vertices}

\begin{definition}
	\label{def.player}
	A generic player in $U \cup V$ will be denoted by $q$. We will say that $q$ is:
	\begin{enumerate}
		\item {\em essential} if $q$ is matched in every maximum weight matching in $G$.
		\item {\em viable} if there is a maximum weight matching $M$ such that $q$ is matched in $M$ and another, $M'$ such that $q$ is not matched in $M'$. 	
		\item {\em subpar} if for every maximum weight matching $M$ in $G$, $q$ is not matched in $M$. 	
		\end{enumerate}
\end{definition}

\begin{definition}
\label{def.player-paid}
	Let $y$ be an imputation in the core. We will say that $q$ {\em gets paid in $y$} if $y_q > 0$ and {\em does not get paid} otherwise. Furthermore, $q$ is {\em paid sometimes} if there is at least one imputation in the core under which $q$ gets paid, and it is {\em never paid} if it is not paid under every imputation. 
\end{definition}

\begin{theorem}
	\label{thm.vertices}
	 For every player $q \in (U \cup V)$: 
		\[ q \ \mbox{is paid sometimes}  \ \iff \ q \ \mbox{is essential} \]  
\end{theorem}
	
\begin{proof}
The proof follows by applying complementary slackness conditions and strict complementarity to the primal LP (\ref{eq.core-primal-bipartite}) and dual LP (\ref{eq.core-dual-bipartite}); see \cite{Sch-book} for formal statements of these facts. By Theorem \ref{thm.SS}, talking about imputations in the core of the assignment game is equivalent to talking about optimal solutions to the dual LP.

 Let $x$ and $y$ be optimal solutions to LP (\ref{eq.core-primal-bipartite}) and LP (\ref{eq.core-dual-bipartite}), respectively. By the Complementary Slackness Theorem, for each $q \in (U \cup V): \ y_q \cdot (x(\delta(q)) - 1) = 0$. 

$(\Rightarrow)$  Suppose $q$ is paid sometimes. Then, there is an optimal solution to the dual LP, say $y$, such that $y_q > 0$. By the Complementary Slackness Theorem, for any optimal solution, $x$, to LP (\ref{eq.core-primal-bipartite}), $x(\delta(q)) = 1$, i.e., $q$ is matched in $x$. Varying $x$ over all optimal primal solutions, we get that $q$ is always matched. In particular, $q$ is matched in all optimal assignments, i.e., integral optimal primal solutions, and is therefore essential. This proves the forward direction.

$(\Leftarrow)$ Strict complementarity implies that corresponding to each player $q$, there is a pair of optimal primal and dual solutions, say $x$ and $y$, such that either $y_q = 0$ or $x(\delta(q)) = 1$ but not both. Assume that $q$ is essential, i.e., it is matched in every integral optimal primal solution. 
 
 We will use Corollary \ref{cor.Birkhoff}, which implies that every fractional optimal primal solution to LP (\ref{eq.core-primal-bipartite}) is a convex combination of integral optimal primal solutions. Therefore $q$ is fully matched in every optimal solution, $x$, to LP (\ref{eq.core-primal-bipartite}), i.e., $x(\delta(q)) = 1$. Therefore there must be an optimal dual solution $y$ such that $y_q > 0$. Hence $q$ is paid sometimes, proving the reverse direction. 
\end{proof}

Theorem \ref{thm.vertices} is equivalent to the following. For every player $q \in (U \cup V)$: 
		\[ q \ \mbox{is never paid} \ \iff \ q \ \mbox{is not essential} \]  
		
Thus core imputations pay only essential players and each of them is paid in some core imputation.  Since we have assumed that the weight of each edge is positive, so is the worth of the game, and all of it goes to essential players. Hence we get:

\begin{corollary}
	\label{cor.vertices}
	In the assignment game, the set of essential players is non-empty and in every core imputation, the entire worth of the game is distributed among essential players; moreover, each of them is paid in some core imputation. 
\end{corollary} 

\begin{remark}
Theorem \ref{thm.SS} and Corollary \ref{cor.vertices} are of much consequence. 
\begin{enumerate}
	\item Corollary \ref{cor.vertices} reveals the following surprising fact: the set of players who are allocated profits in a core imputation is {\em independent} of the set of teams that play. 
	\item The identification of these players, and the exact manner in which the total profit is divided among them, follows the negotiating process described on Section \ref{sec.negotiation}, in which each player ascertains his/her negotiating power based on all possible sub-coalitions he/she participates in. In turn, this process identifies                                                              agents who play in {\em all} possible maximum weight matchings. 
	\item Perhaps the most remarkable aspect of Theorem \ref{thm.SS} is that each possible outcome of this very real process is captured by an inanimate object, namely an optimal solution to the dual, LP (\ref{eq.core-dual-bipartite}). 
\end{enumerate}
\end{remark}

By Corollary \ref{cor.vertices}, core imputations reward only essential players. This raises the following question: Can't a non-essential player, say $q$, team up with another player, say $p$, and secede, by promising $p$ almost all of the resulting profit? The answer is ``No'', because the dual (\ref{eq.core-dual-bipartite}) has the constraint $y_q + y_p \geq w_{qp}$. Therefore, if $y_q = 0$, $y_p \geq w_{q p}$, i.e., $p$ will not gain by seceding together with $q$.

%% file: Definitions.tex
\subsection{Definitions and Preliminary Facts}
\label{sec.matching-game}

The {\em assignment game}, $G = (U, V, E), \ w: E \rightarrow \cR_+$, has been defined in the Introduction. We start by giving definitions needed to state the Shapley-Shubik Theorem. 

\begin{definition}
	\label{sec.coalition}
		The set of all players, $U \cup V$, is called the {\em grand coalition}. A subset of the players, $(S_u \cup S_v)$, with $S_u \subseteq U$ and $S_v \subseteq V$, is called a {\em coalition} or a {\em sub-coalition}.
\end{definition}

\begin{definition}
	\label{def.worth}
	The {\em worth} of a coalition $(S_u \cup S_v)$ is defined to be the maximum profit that can be generated by teams within $(S_u \cup S_v)$ and is denoted by $p(S_u \cup S_v)$. Formally, $p(S_u \cup S_v)$ is the weight of a maximum weight matching in the graph $G$ restricted to vertices in $(S_u \cup S_v)$ only. $p(U \cup V)$ is called the {\em worth of the game}. The {\em characteristic function} of the game is defined to be $p: 2^{U \cup V} \rightarrow \cR_+$.   
\end{definition}

\begin{definition}
	\label{def.imputation}	
	An {\em imputation}\footnote{Some authors prefer to call this a pre-imputation, while using the term imputation when individual rationality is also satisfied.} gives a way of dividing the worth of the game, $p(U \cup V)$, among the agents. It consists of two functions $u: {U} \rightarrow \cR_+$ and $v: {V} \rightarrow \cR_+$ such that $\sum_{i \in U} {u(i)} + \sum_{j \in V} {v(j)} = p(U \cup V)$. 
\end{definition}
	
\begin{definition}
	\label{def.core}
	An imputation $(u, v)$ is said to be in the {\em core of the assignment game} if for any coalition $(S_u \cup S_v)$, the total worth allocated to agents in the coalition is at least as large as the worth that they can generate by themselves, i.e., $\sum_{i \in S_u} {u(i)} +  \sum_{j \in S_v} {v(j)} \geq p(S)$.
\end{definition}
 
We next describe the characterization of the core of the assignment game given by Shapley and Shubik \cite{Shapley1971assignment}\footnote{Shapley and Shubik had described this game in the context of the housing market in which agents are of two types, buyers and sellers. They had shown that each imputation in the core of this game gives rise to unique prices for all the houses. In this paper we will present the assignment game in a variant of the tennis setting given in the Introduction; this will obviate the need to define ``prices'', hence leading to simplicity.}. 

As stated in Definition \ref{def.worth}, the worth of the game, $G = (U, V, E), \ w: E \rightarrow \cR_+$, is the weight of a maximum weight matching in $G$. Linear program (\ref{eq.core-primal-bipartite}) gives the LP-relaxation of the problem of finding such a matching. In this program, variable $x_{ij}$ indicates the extent to which edge $(i, j)$ is picked in the solution. Matching theory tells us that this LP always has an integral optimal solution \cite{LP.book}; the latter is a maximum weight matching in $G$.

	\begin{maxi}
		{} {\sum_{(i, j) \in E}  {w_{ij} x_{ij}}}
			{\label{eq.core-primal-bipartite}}
		{}
		\addConstraint{\sum_{(i, j) \in E} {x_{ij}}}{\leq 1 \quad}{\forall i \in U}
		\addConstraint{\sum_{(i, j) \in E} {x_{ij}}}{\leq 1 }{\forall j \in V}
		\addConstraint{x_{ij}}{\geq 0}{\forall (i, j) \in E}
	\end{maxi}

Taking $u_i$ and $v_j$ to be the dual variables for the first and second constraints of (\ref{eq.core-primal-bipartite}), we obtain the dual LP: 

 	\begin{mini}
		{} {\sum_{i \in U}  {u_{i}} + \sum_{j \in V} {v_j}} 
			{\label{eq.core-dual-bipartite}}
		{}
		\addConstraint{ u_i + v_j}{ \geq w_{ij} \quad }{\forall (i, j) \in E}
		\addConstraint{u_{i}}{\geq 0}{\forall i \in U}
		\addConstraint{v_{j}}{\geq 0}{\forall j \in V}
	\end{mini}

\begin{theorem}
	\label{thm.SS}
	(Shapley and Shubik \cite{Shapley1971assignment})
	The imputation $(u, v)$ is in the core of the assignment game if and only if it is an optimal solution to the dual LP, (\ref{eq.core-dual-bipartite}). 
\end{theorem}

By Theorem \ref{thm.SS}, the core of the assignment game is a convex polyhedron. Shapley and Shubik shed further light on the structure of the core by showing that it has two special imputations which are furtherest apart and so can be thought of as antipodal imputations. In the tennis club setup, one of these imputations maximizes the earnings of women players and the other maximizes the earnings of men players. These are illustrated in Examples \ref{ex.3} and \ref{ex.4}. 

Finally, we state a fundamental fact about LP (\ref{eq.core-primal-bipartite}); its corollary will be used in a crucial way in Theorems \ref{thm.vertices} and \ref{thm.edges}.

\begin{theorem}
	\label{thm.Birkhoff}
	(Birkhoff \cite{Birkhoff1946three})	The vertices of the polytope defined by the constraints of LP (\ref{eq.core-primal-bipartite}) are $0/1$ vectors, i.e., they are matchings in $G$.
\end{theorem}

\begin{corollary}
	\label{cor.Birkhoff}
	Any fractional matching in a bipartite graph is a convex combination of integral matchings. 
\end{corollary}

%% file: Edges.tex
\subsection{The second question: Allocations made to teams by core imputations}
\label{sec.edges}

\begin{definition}
	\label{def.team}
	By a {\em mixed doubles team} we mean an edge in $G$; a generic one will be denoted as $e = (u, v)$. We will say that $e$ is:
	\begin{enumerate}
		\item {\em essential} if $e$ is matched in every maximum weight matching in $G$.
		\item {\em viable} if there is a maximum weight matching $M$ such that $e \in M$, and another, $M'$ such that $e \notin M'$. 
		\item {\em subpar} if for every maximum weight matching $M$ in $G$, $e \notin M$. 
	\end{enumerate}
	\end{definition}
	
\begin{definition}
\label{def.team-paid}
	 Let $y$ be an imputation in the core of the game. We will say that $e$ is {\em fairly paid in $y$} if $y_u + y_v = w_e$ and it is {\em overpaid} if $y_u + y_v > w_e$\footnote{Observe that by the first constraint of the dual LP (\ref{eq.core-dual-bipartite}), these are the only possibilities.}. Finally, we will say that $e$ is {\em always paid fairly} if it is fairly paid in every imputation in the core.
\end{definition}

\begin{theorem}
	\label{thm.edges}
	 For every team $e \in E$: 
		\[ e \ \mbox{is always paid fairly} \ \iff \ e \ \mbox{is viable or essential} \]
\end{theorem}
	
\begin{proof}
The proof is similar to that of Theorem \ref{thm.vertices}. Let $x$ and $y$ be optimal solutions to LP (\ref{eq.core-primal-bipartite}) and LP (\ref{eq.core-dual-bipartite}), respectively. By the Complementary Slackness Theorem, for each $e = (u, v) \in E: \ \ x_e \cdot (y_u + y_v - w_e) = 0$.

$(\Leftarrow)$ To prove the reverse direction, suppose $e$ is viable or essential. Then there is an optimal solution to the primal, say $x$, under which it is matched. Therefore,  $x_e > 0$. Let $y$ be an arbitrary optimal dual solution. Then, by the Complementary Slackness Theorem, $y_u + y_v = w_e$, i.e., $e$ is fairly paid in $y$. Varying $y$ over all optimal dual solutions, we get that $e$ is always paid fairly. 

$(\Rightarrow)$ To prove the forward direction, we will use strict complementarity. It implies that corresponding to each team $e$, there is a pair of optimal primal and dual solutions $x$ and $y$ such that either $x_e = 0$ or $y_u + y_v = w_e$ but not both. 

Assume that team $e$ is always fairly paid, i.e., under every optimal dual solution $y$, $y_u + y_v = w_e$. By strict complementarity, there must be an optimal primal solution $x$ for which $x_e > 0$. Now corollary \ref{cor.Birkhoff} implies that $x$ is a convex combination of optimal assignments. Therefore, there must be an optimal assignment in which $e$ is matched. Therefore $e$ is viable or essential and the forward direction also holds. 
\end{proof}

Negating both sides of the implication proved in Theorem \ref{thm.edges}, we get the following  implication. For every team $e \in E$: 
		\[ e \ \mbox{is subpar} \ \iff \ e \ \mbox{is sometimes overpaid} \]

Clearly, this statement is equivalent to the statement proved Theorem \ref{thm.edges} and hence contains no new information. However, it provides a new viewpoint. These two equivalent  statements yield the following assertion, which at first sight seems incongruous with what we desire from the notion of the core and the just manner in which it allocates profits:

\begin{center}
{\em Whereas viable and essential teams are always paid fairly, subpar teams are sometimes overpaid.}
\end{center}

How can the core favor subpar teams over viable and essential teams? An explanation is provided in the Introduction, namely a subpar team $(i, j)$ gets overpaid because $i$ and $j$ create worth by playing in competent teams with other players. 

Finally, we observe that contrary to Corollary \ref{cor.vertices}, which says that the set of essential players is non-empty, the set of essential teams may be empty, as is the case in Examples \ref{ex.1} and \ref{ex.2} in Section \ref{sec.negotiation}.


\subsection{The third question: Degeneracy}
\label{sec.degeneracy}

 Next we use Theorems \ref{thm.vertices} and \ref{thm.edges} to get insights into degeneracy. Clearly, if an assignment game is non-degenerate, then every team and every player is either always matched or always unmatched in the set of maximum weight matchings in $G$, i.e., there are no viable teams or players. Since viable teams and players arise due to degeneracy, in order to understand the phenomenon of degeneracy, we need to understand how viable teams and players behave with respect to core imputations; this is done in the next corollary.
 
\begin{corollary}
	\label{cor.degen}
	In the presence of degeneracy, imputations in the core of an assignment game treat:
	\begin{itemize}
			\item  viable players in the same way as subpar players, namely they are never paid.  
		\item viable teams in the same way as essential teams, namely they are always fairly paid. 
	\end{itemize}
\end{corollary}

%% file: Negotiate.tex
\section{Insights Provided by the Core into Negotiating Power of Agents}
\label{sec.negotiation}


\begin{figure}[h]
\begin{center}
\includegraphics[width=2.4in]{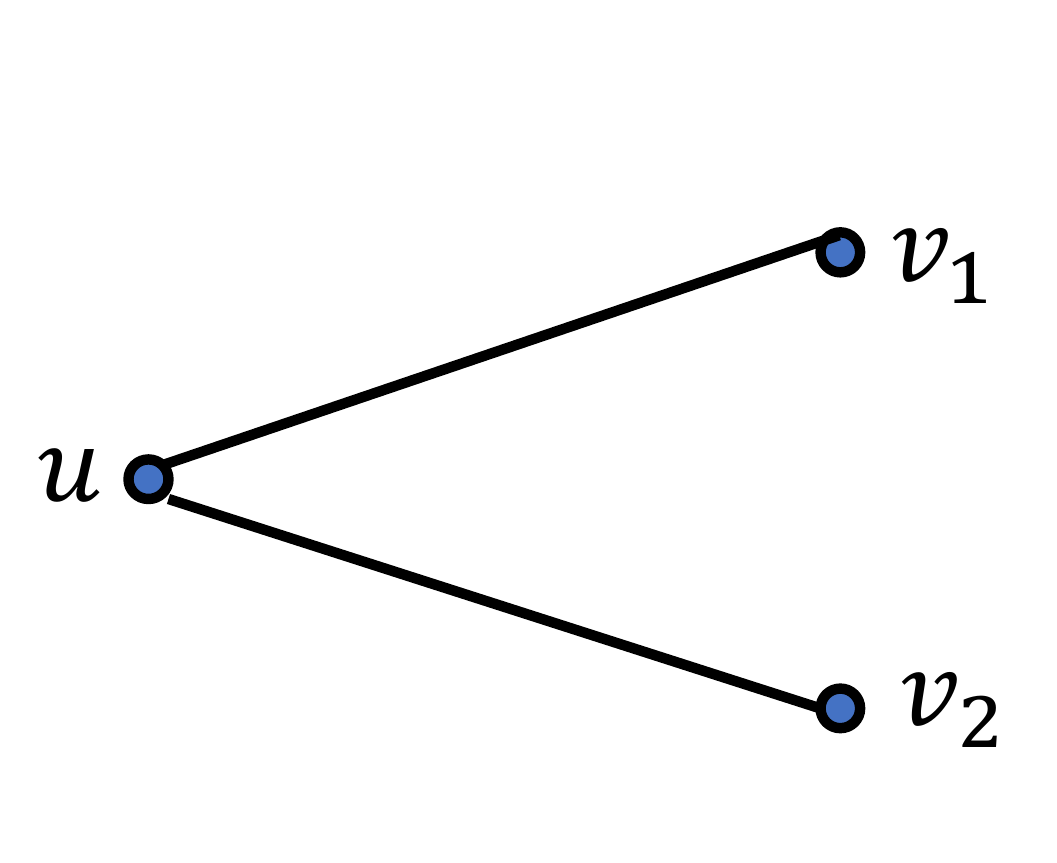}
\caption{The graph for Example \ref{ex.1}}
\label{fig.1}
\end{center}
\end{figure}


\begin{example}
	\label{ex.1}
	Consider an assignment game whose bipartite graph, on the three agents $u, v_1, v_2$, and two edges, is given in Figure \ref{fig.1}. Clearly, one of $v_1$ and $v_2$ will be left out in any matching. First assume that the weight of both edges is 1. If so, the unique imputation in the core gives zero to $v_1$ and $v_2$, and 1 to $u$. Next assume that the weights of the edges $(u, v_1)$ and $(u, v_2)$ are 1 and $1 + \epsilon$ respectively, for a small $\epsilon > 0$. If so, the unique imputation in the core gives $0, \epsilon$ and $1$ to $v_1$, $v_2$ and $u$, respectively.  
\end{example}

How fair are the imputations given in Example \ref{ex.1}? As stated in the Introduction, imputations in the core have a lot to do with the negotiating power of individuals and sub-coalitions. Let us argue that when the imputations given above are viewed from this angle, they are fair in that the profit allocated to an agent is consistent with their negotiating power, i.e., their worth. In the first case, whereas $u$ has alternatives, $v_1$ and $v_2$ don't. As a result, $u$ will squeeze out all profits from whoever she plays with, by threatening to partner with the other player. Therefore $v_1$ and $v_2$ have to be content with no rewards! In the second case, $u$ can always threaten to match up with $v_2$. Therefore $v_1$ has to be content with a profit of $\epsilon$ only.

\begin{example}
	\label{ex.2}
	 Consider an assignment game whose bipartite graph, shown in Figure \ref{fig.2}, has four edges, $(u_1, v_1), (u_1, v_2)$, $(u_2, v_2), (u_2, v_3)$ on the five agents $u_1, u_2, v_1, v_2, v_3$. Let the wights of these four edges be $1, 1.1, 1.1$ and $1$, respectively. This game has two different maximum weight matchings, each of weight 2.1, which is also the worth of this game. 
\end{example}


\begin{figure}[h]
\begin{center}
\includegraphics[width=2.4in]{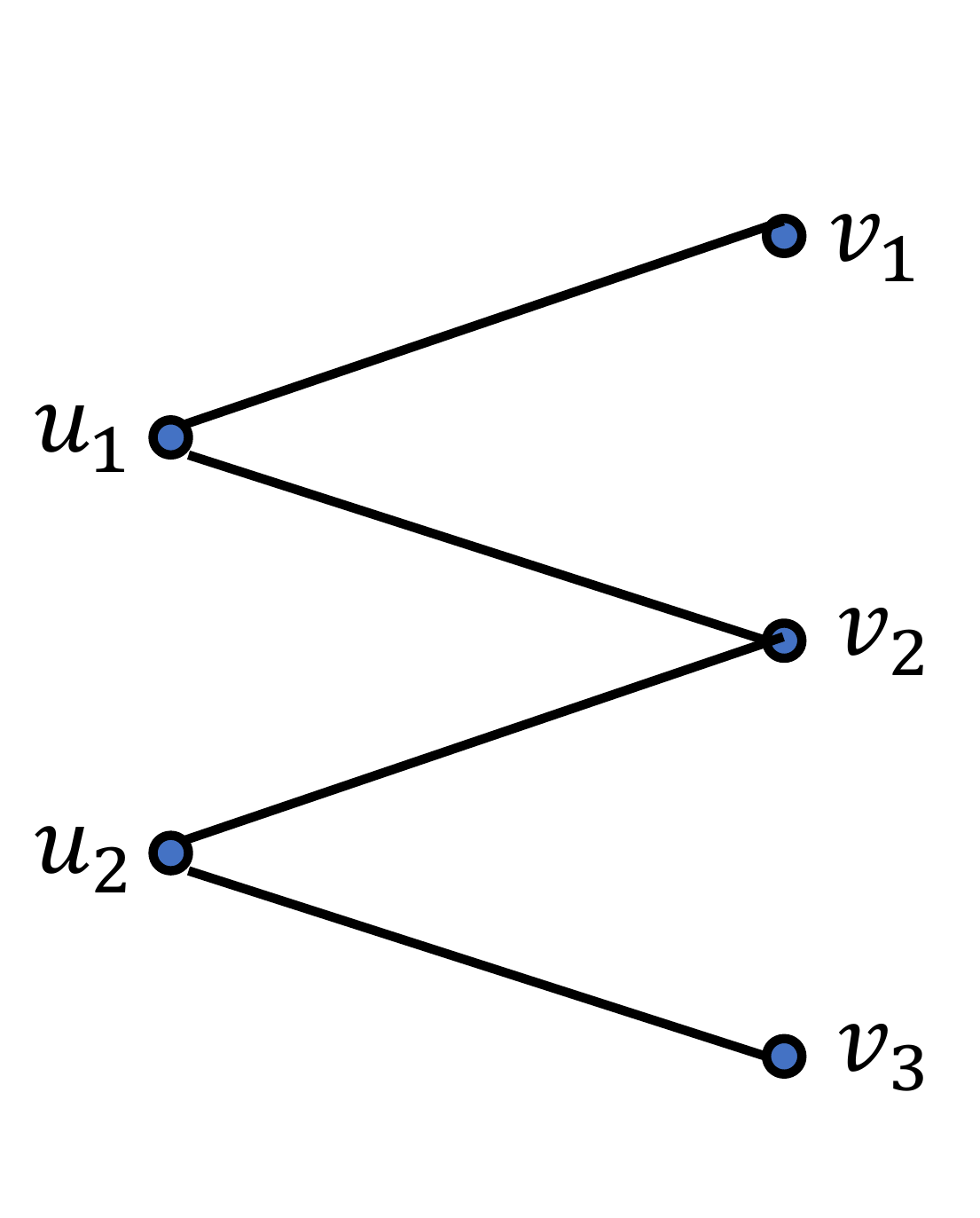}
\caption{The graph for Example \ref{ex.2}.}
\label{fig.2}
\end{center}
\end{figure}


 In Example \ref{ex.2}, at first sight, $v_2$ looks like the dominant player, since he has two choices of partners, namely $u_1$ and $u_2$, and because teams involving him have the biggest earnings, namely 1.1 as opposed to 1. Yet, the unique core imputation in the core awards $1, 1, 0, 0.1, 0$ to agents $u_1, u_2, v_1, v_2, v_3$, respectively. 
 
 The question arises: ``Why is $v_2$ allocated only $0.1$? Assume that the maximum weight matching chosen is $\{(u_1, v_2), (u_2, v_3) \}$. If $v_2$ tries to negotiate a higher profit than $0.1$ from $u_1$, she will threaten to play with $v_1$, who will play for essentially nothing, since he was left out of the tournament so far. As a result, $u_1$ is guaranteed a profit of essentially 1 and so $v_2$ has to be content with $0.1$. A similar reasoning applies if the second maximum weight matching is chosen. Hence the core imputation has indeed allocated profits according to the negotiating power of each agent. 
 

\begin{example}
	\label{ex.3}
	The game of Figure \ref{fig.3} has five players. Assume that the weights of edges $(u_1, v_1)$, $(u_1, v_2)$ and $(u_2, v_1)$ are 1 each and the weights of $(u_2, v_2)$ and $(u_3, v_2)$ are $0.4$ and $0.9$, respectively. The maximum weight matching picks edges $(u_1, v_2)$ and  $(u_2, v_1)$, having total weight of 2. 
	
	Player $u_1$ is highly competitive and has alternatives, yet the two antipodal points in the core assign to $u_1, u_2, u_3, v_1, v_2$ profits of $0.1, 0.1, 0, 0.9, 0.9$ and $0, 0, 0, 1, 1$ in the woman-optimal and man-optimal imputations, respectively. Observe that $u_1$ is assigned profits of 0 and $0.1$ only in these two imputations. Why can't she negotiate a higher profit from $v_2$? The reason is that $v_2$ has a guaranteed alternative in $u_3$, who is available and with whom he can make a profit of essentially $0.9$. 
\end{example}


\begin{figure}[h]
\begin{center}
\includegraphics[width=2.4in]{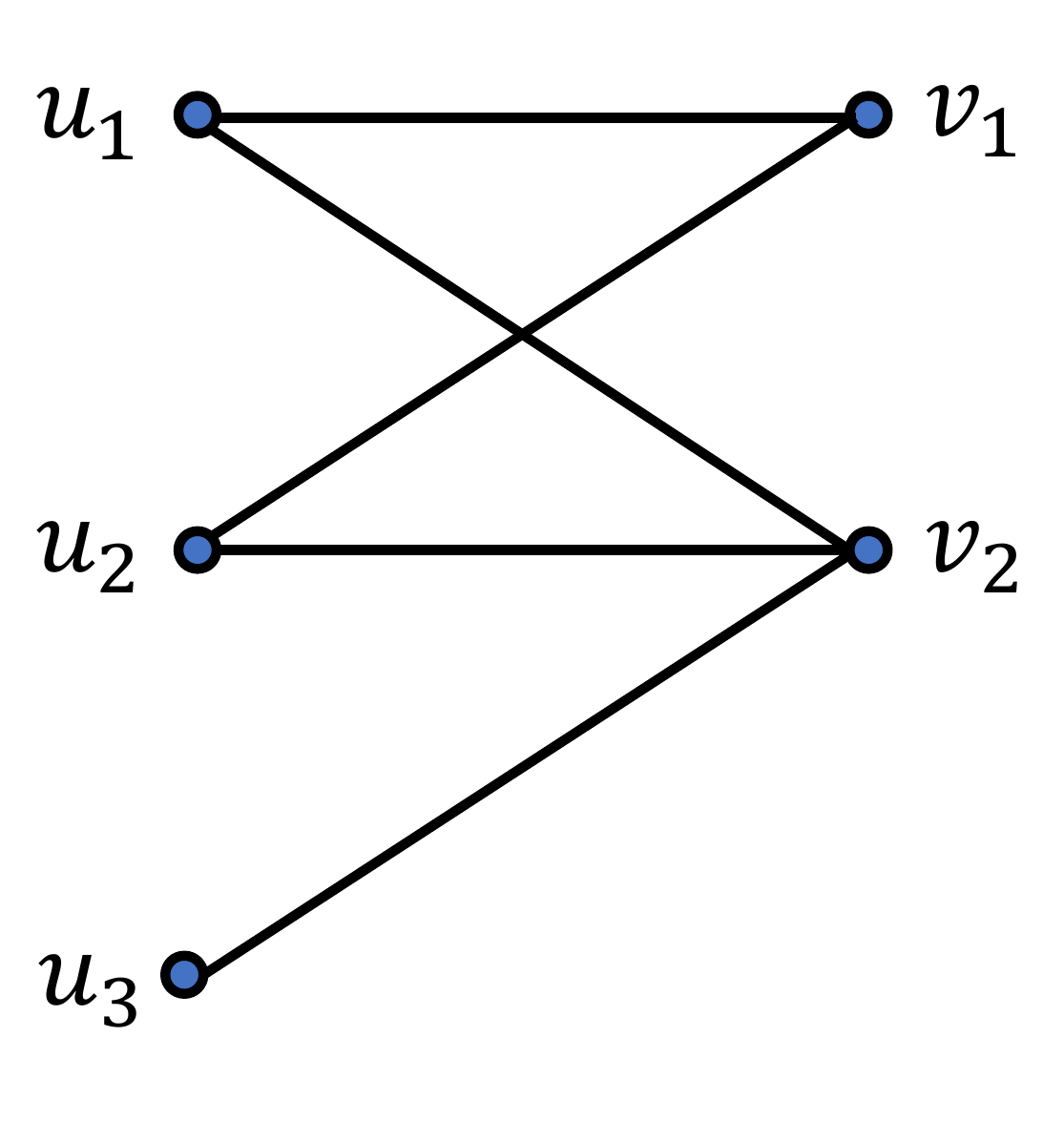}
\caption{The graph for Example \ref{ex.3}.}
\label{fig.3}
\end{center}
\end{figure}


 \begin{example}
 	\label{ex.4}
 	Consider the assignment game whose bipartite graph, in eight vertices, is shown in Figure \ref{fig.4}. Assume that the weights of $(u_1, v_1)$ and $(u_2, v_2)$ are 100 each, the weights of $(u_1, v_3)$ and $(u_2, v_4)$ are 51 each, and the weights of $(u_3, v_2)$ and $(u_4, v_1)$ are 50 each. Clearly, the worth of the game is 202 and is given by matching the last four edges. There are two antipodal imputations in the core. The woman-optimal imputation gives 51 to each of $u_1$ and $u_2$, 50 to each of $v_1$ and $v_2$, and zero to the rest. The man-optimal imputation gives 50 to each of $u_1$, $u_2$, $v_1$ and $v_2$, 1 to each of $v_3$ and $v_4$, and zero to $u_3$ and $u_4$. 
 \end{example}


\begin{figure}[h]
\begin{center}
\includegraphics[width=2.4in]{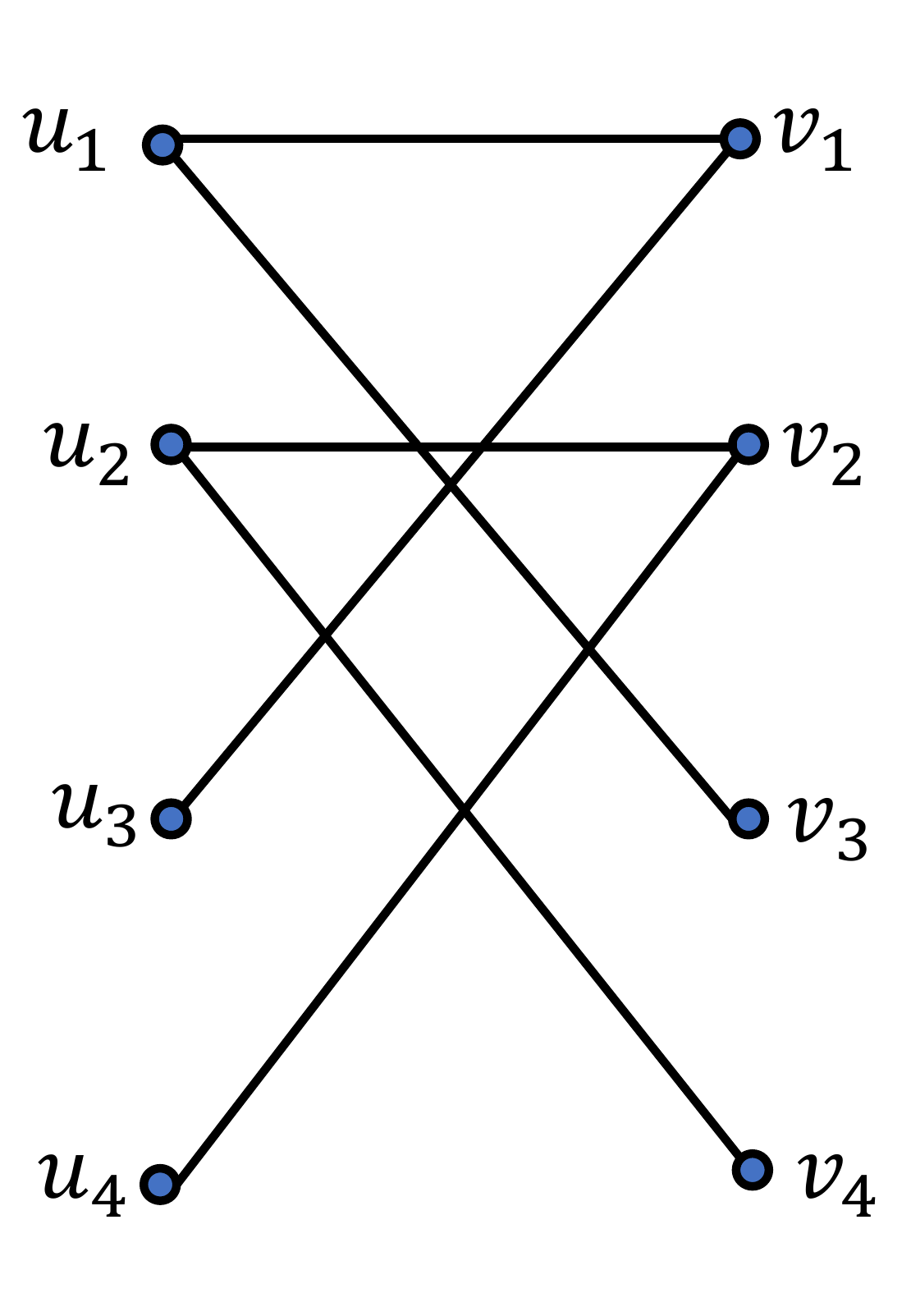}
\caption{The graph for Example \ref{ex.4}.}
\label{fig.4}
\end{center}
\end{figure}


Let us partition the players into two sets: the {\em bottom players} consisting of $u_3, u_4, v_3$ and $v_4$, and the {\em top players} consisting of the rest. The question is why do the bottom  players get so little profit as compared to the top players? In particular, the bottom players  get zero profit in the woman-optimal imputation and a total of only 2 in the man-optimal imputation. 

The reason is that the top players can generate a worth of 200 on their own, via the teams $(u_1, v_1)$ and $(u_2, v_2)$. This gives them the power to negotiate a total profit of at least 200, leaving a profit of at most 2 for the first four players. Indeed, the core imputations must respect this, to prevent the top players from seceding.

%% file: General.tex
\section{The Core of Concurrent Games}  
\label{sec.general}

In Section \ref{sec.gen-matching-game}, we define the general graph matching game. Once again, a good way of introducing this game is in the context of a tennis club that needs to enter doubles teams in a tournament; however, this time the club has players of one gender only, and so any two players can form a doubles team and the underlying graph is, in general, non-bipartite.   
 
 Deng et al. \cite{Deng1999algorithms} showed that the core of this game is non-empty if and only if the weights of maximum weight integral and fractional matchings concur; if so, we will say that the game is {\em concurrent}. In Section \ref{sec.gen-matching-game}, we give the LP-relaxation of the problem of finding a maximum weight fractional matching and thereby provide the underlying reason for the above-stated result of \cite{Deng1999algorithms} as well as their characterization of the core of concurrent games.

\subsection{Definitions and Preliminary Facts} 
\label{sec.gen-matching-game}

\begin{definition}
	\label{def.matching-game}
	The {\em general graph matching game} consists of an undirected graph $G = (V, E)$ and an edge-weight function $w$. The vertices $i \in V$ are tennis players and an edge $(i, j) \in E$ represents the fact that players $i$ and $j$ are happy to form a doubles team; if so, $w_{i j}$ represents the profit generated by this team.  
\end{definition}     

\begin{definition}
	\label{def.worth}
	The {\em worth} of a coalition $S \subseteq V$ is defined to be the maximum profit that can be generated by teams within $S$ and is denoted by $p(S)$. Formally, $p(S)$ is the weight of a maximum weight matching in the graph $G$ restricted to vertices in $S$ only. The {\em worth of the game} is defined to be $p(V)$, i.e., the worth of the {\em grand coalition, $V$}.  The {\em characteristic function} of the game is defined to be $p: 2^{V} \rightarrow \cR_+$.   
\end{definition}

\begin{definition}
	\label{def.imputation}	
	An {\em imputation}\footnote{Some authors prefer to call this a pre-imputation, while using the term imputation when individual rationality is also satisfied.} gives a way of dividing the worth of the game, $p(V)$, among the agents. Formally, it is a function $v: {V} \rightarrow \cR_+$ such that $\sum_{i \in V} {v(i)} = p(V)$. 
\end{definition}
	
\begin{definition}
	\label{def.core}
	An imputation $v$ is said to be in the {\em core of the matching game} if for any coalition $S \subseteq V$, the total worth allocated to agents in $S$ is at least as large as the worth that they can generate by themselves, i.e., $v(S) \geq p(S)$, where $v(S) = \sum_{i \in S} {v(i)}$.
\end{definition}

We will work with the following LP (\ref{eq.core-primal}), whose optimal solutions are maximum weight fractional matchings in $G$.

	\begin{maxi}
		{} {\sum_{(i, j) \in E}  {w_{ij} x_{ij}}}
			{\label{eq.core-primal}}
		{}
		\addConstraint{\sum_{(i, j) \in E} {x_{ij}}}{\leq 1 \quad}{\forall i \in V}
		\addConstraint{x_{ij}}{\geq 0}{\forall (i, j) \in E}
	\end{maxi}
	
Note that in case $G$ is bipartite, LP (\ref{eq.core-primal}) is equivalent to LP (\ref{eq.core-primal-bipartite}). Therefore, by Theorem \ref{thm.Birkhoff}, it always has an integral optimal solution. On the other hand, if $G$ is non-bipartite, LP (\ref{eq.core-primal}) may have no integral optimal solutions, e.g., a triangle with unit weight edges. However, by Theorem \ref{thm.Balinski}, this LP always has a half-integral optimal solution.

\begin{theorem}
\label{thm.Balinski}  (Balinski \cite{Balinski1965integer}) 
	For a general graph, the vertices of the polytope defined by the constraints of LP (\ref{eq.core-primal}) are half-integral, such that the edges set to 1 form a matching and those set to half form disjoint odd-length cycles. 
\end{theorem}

Taking $v_i$ to be dual variables for the first constraint of (\ref{eq.core-primal}), we obtain  LP (\ref{eq.core-dual}). Any feasible solution to this LP is called a {\em cover} of $G$, since for each edge $(i, j)$, $v_i$ and $v_j$ cover edge $(i, j)$ in the sense that $v_i + v_j \geq w_{ij}$. An optimal solution to this LP is a {\em minimum cover}. We will say that $v_i$ is the {\em profit} of player $i$.

 	\begin{mini}
		{} {\sum_{i \in V}  {v_{i}}} 
			{\label{eq.core-dual}}
		{}
		\addConstraint{v_i + v_j}{ \geq w_{ij} \quad }{\forall (i, j) \in E}
		\addConstraint{v_{i}}{\geq 0}{\forall i \in V}
	\end{mini}
	
Let $Q_f$ ($Q_i$) be the weight of a maximum weight fractional (integral) matching in $G$. Now, $Q_f \geq Q_i$, since every integral matching is also a fractional matching. By the LP Duality Theorem, $Q_f$ equals the total value of a minimum cover. On the other hand, $Q_i$ is the worth of the game.

\cite{Deng1999algorithms} proved that the core of the general graph matching game is non-empty if and only if $Q_f = Q_i$. If so, by a proof that is similar to that of the Shapley-Shubik Theorem, it is easy to see each optimal solution to the dual LP, namely LP (\ref{eq.core-dual}),  gives a way for distributing the worth of the game among agents so that the condition of the core is satisfied, i.e., it is a core imputation. The converse is also true, i.e., every core imputation is an optimal solution to the dual LP. We summarize below.

\begin{theorem}
	\label{thm.Deng}
	(Deng et al. \cite{Deng1999algorithms}) 
	The core of the general graph matching game is non-empty if and only if $Q_f = Q_i$. If so, the set of core imputation is precisely the set of optimal solutions to the dual LP, LP (\ref{eq.core-dual}). 
\end{theorem}

\begin{example}
	\label{ex.5}
	Consider the graph given in Figure \ref{fig.5}. Assume that the weight of edge $(v_2, v_7)$ is 2 and the weights of the rest of the edges is 1. This graph has three maximum weight integral matching, of weight 4, and edge $(v_2, v_7)$ is in each of them; in addition, the three matchings contain the edge sets $\{(v_1, v_6), (v_4, v_5) \}, \{(v_1, v_6), (v_3, v_4) \}$ and $\{(v_3, v_4), (v_5, v_6) \}$. This graph also has three fractional matchings having a weight of 4, which are not integral: the seven-cycle $v_1, v_2, v_7, v_3, v_4, v_5, v_6$ taken half-integrally; the three-cycle $v_1, v_2, v_7$ taken half-integrally together with the edges $\{(v_3, v_4), (v_5, v_6)\}$; and the three-cycle $v_2, v_3, v_7$ taken half-integrally together with the edges $\{(v_4, v_5), (v_1, v_6)\}$. 
	
	By Theorem \ref{thm.Deng}, this game has a non-empty core. The unique core imputation assigns a profit of 1 to each of $v_2, v_4, v_6, v_7$ and zero to the rest. Observe that if the weight of edge $(v_2, v_7)$ is decreased a bit, then the core will become empty and if it is increased a bit, then the maximum weight fractional matchings will all be integral and the core will remain non-empty. 
\end{example}


\begin{figure}[h]
\begin{center}
\includegraphics[width=2.8in]{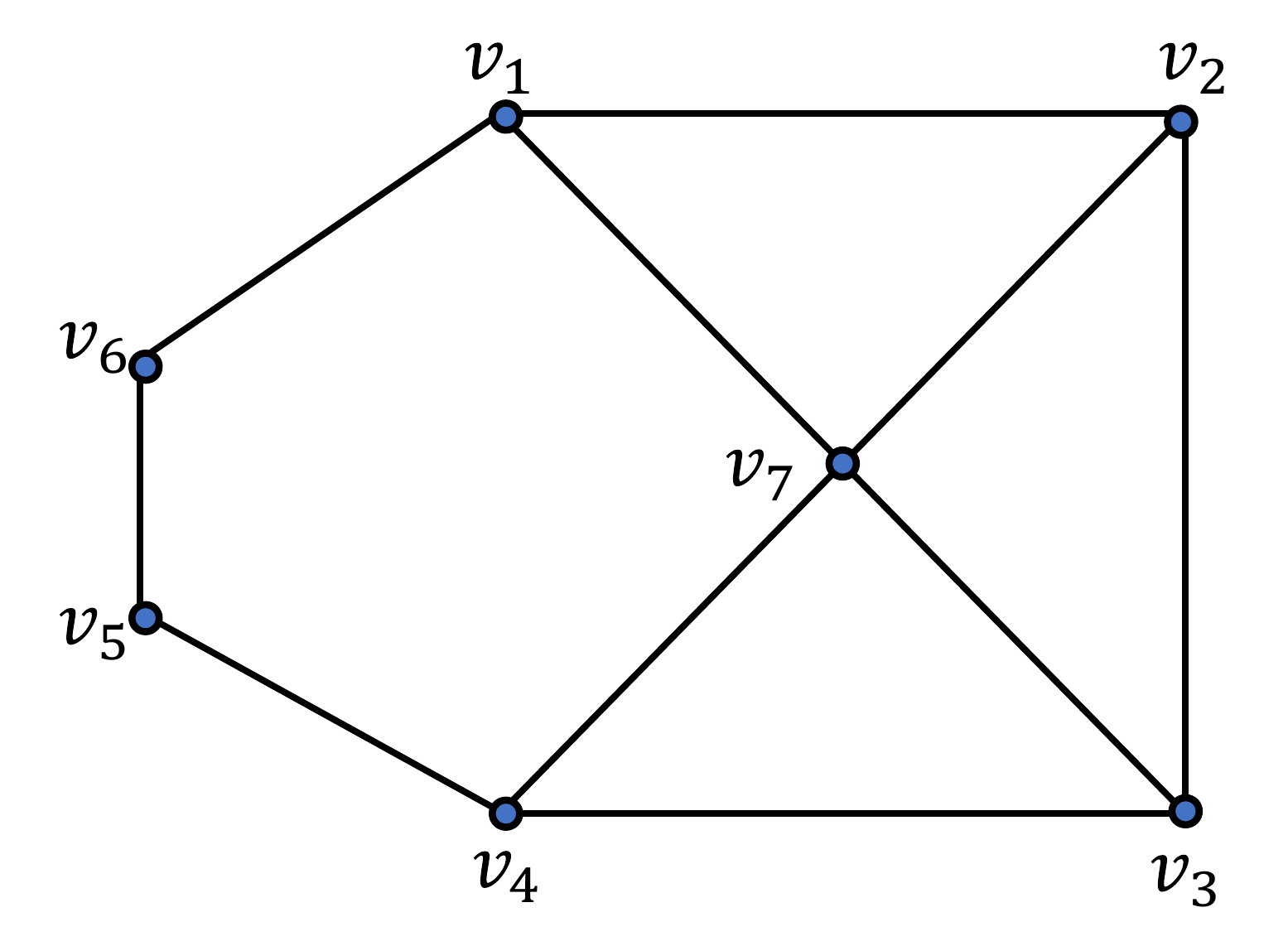}
\caption{The graph for Example \ref{ex.5}.}
\label{fig.5}
\end{center}
\end{figure}



For completeness, we describe a different LP for general graphs, namely LP (\ref{eq.core-primal-general}), given by Edmonds \cite{Edmonds.matching}, which always has integral optimal solutions; these are maximum weight matchings in the graph. This LP is an enhancement of LP (\ref{eq.core-primal}) via odd set constraints, as specified in (\ref{eq.core-primal-general}). The latter constraints are exponential in number, namely one for every odd subset $S$ of vertices. Clearly, the total number of integral matched edges in this set can be at most ${{(|S|-1)} \over 2}$. The constraint imposes this bound on any fractional matching as well, thereby ensuring that a fractional matching that picks each edge of an odd cycle half-integrally is disallowed and there is always an integral optimal matching. 

	\begin{maxi}
		{} {\sum_{(i, j) \in E}  {w_{ij} x_{ij}}}
			{\label{eq.core-primal-general}}
		{}
		\addConstraint{\sum_{(i, j) \in E} {x_{ij}}}{\leq 1 \quad}{\forall i \in V}
		\addConstraint{\sum_{(i, j) \in S} {x_{ij}}}{\leq {{(|S|-1)} \over 2} \quad}{\forall S \subseteq V, \ S \ \mbox{odd}}
		\addConstraint{x_{ij}}{\geq 0}{\forall (i, j) \in E}
	\end{maxi}

The dual of this LP has, in addition to variables corresponding to vertices, $v_i$, exponentially many more variables corresponding to odd sets, $z_S$, as given in (\ref{eq.core-dual-general}). As a result, the entire worth of the game does not reside on vertices only --- it also resides on odd sets.

 	\begin{mini}
		{} {\sum_{i \in V}  {v_{i}} + \sum_{S \subseteq V, \ \mbox{odd}} {z_S}} 
			{\label{eq.core-dual-general}}
		{}
		\addConstraint{v_i + v_j + \sum_{S \, \ni \, i, j}{z_S}}{ \geq w_{ij} \quad }{\forall (i, j) \in E}
		\addConstraint{v_{i}}{\geq 0}{\forall i \in V}
		\addConstraint{z_{S}}{\geq 0}{\forall S \subseteq V, \ S \ \mbox{odd}}
	\end{mini}

There is no natural way of dividing $z_S$ among the vertices in $S$ to restore core properties. The situation is more serious than that: it turns out that in general, the core of a non-bipartite game may be empty. 

This is easy to see for the graph $K_3$, i.e., a clique on three vertices, $i, j, k$, with a weight of 1 on each edge. Any maximum matching in $K_3$ has only one edge, and therefore the worth of this game is 1. Suppose there is an imputation $v$ which lies in the core. Consider all three two-agent coalitions. Then, we must have:
\[ v(i) + v(j) \geq 1, \ \ \ \  v(j) + v(k) \geq 1 \ \ \ \ \mbox{and} \ \ \ \ v(i) + v(k) \geq 1 .\]
This implies $v(i) + v(j) + v(k) \geq 3/2$ which exceeds the worth of the game, giving a contradiction.

One recourse to this eventuality was provided in \cite{Va.general}, which gives a ${1 \over 3}$-approximate core imputation for such games, i.e., the weight of a maximum weight matching in the graph is distributed among vertices in such a way that the total profit accrued by agents in a sub-coalition $S \subseteq V$ is at least ${1 \over 3}$ fraction of the profit which $S$ can generate by itself.

%% file: Vertices2.tex
\subsection{Answers to the Three Questions for Concurrent Games}
\label{sec.Gen-Insights}

In this section, we will assume that $G = (V, E)$, $w: E \rightarrow \QQ_+$ is a concurrent game, and we will provide answers to the three issues raised in the Introduction. We first provide the necessary definitions, which are adaptations of definitions from Section \ref{sec.matching-game}.

\begin{definition}
	\label{def.Gen-player}
	A generic player in $V$ will be denoted by $q$. We will say that $q$ is:
	\begin{enumerate}
		\item {\em essential} if $q$ is matched in every maximum weight integral matching in $G$.
		\item {\em viable} if there is a maximum weight integral matching $M$ such that $q$ is matched in $M$ and another, $M'$ such that $q$ is not matched in $M'$. 	
		\item {\em subpar} if for every maximum weight integral matching $M$ in $G$, $q$ is not matched in $M$. 	
		\end{enumerate}
\end{definition}

\begin{definition}
\label{def.Gen-player-paid}
	Let $y$ be an imputation in the core. We will say that $q$ {\em gets paid in $y$} if $y_q > 0$ and {\em does not get paid} otherwise. Furthermore, $q$ is {\em paid sometimes} if there is at least one imputation in the core under which $q$ gets paid, and it is {\em never paid} if it is not paid under every imputation. 
\end{definition}

\begin{definition}
	\label{def.Gen-team}
	By a {\em team} we mean an edge in $G$; a generic one will be denoted as $e = (u, v)$. We will say that $e$ is:
	\begin{enumerate}
		\item {\em essential} if $e$ is matched in every maximum weight matching in $G$.
		\item {\em viable} if there is a maximum weight matching $M$ such that $e \in M$, and another, $M'$ such that $e \notin M'$. 
		\item {\em subpar} if for every maximum weight matching $M$ in $G$, $e \notin M$. 
	\end{enumerate}
	\end{definition}
	
\begin{definition}
\label{def.Gen-team-paid}
	 Let $y$ be an imputation in the core of a concurrent game. We will say that $e$ is {\em fairly paid in $y$} if $y_u + y_v = w_e$ and it is {\em overpaid} if $y_u + y_v > w_e$\footnote{Observe that by the first constraint of the dual LP (\ref{eq.core-dual-bipartite}), these are the only possibilities.}. Finally, we will say that $e$ is {\em always paid fairly} if it is fairly paid in every imputation in the core, and it is {\em sometimes overpaid} if there is a core imputation in which it is overpaid.
\end{definition}

\begin{theorem}
	\label{thm.Gen-insights}
	The following hold:
\begin{enumerate}
	\item For every player $q \in (U \cup V)$: 
		\[ q \ \mbox{is paid sometimes}  \ \Rightarrow \ q \ \mbox{is essential} \]
	\item 	 For every team $e \in E$: 
		\[ e \ \mbox{is viable or essential} \ \Rightarrow \ e \ \mbox{is always paid fairly} \]
\end{enumerate}  
\end{theorem}
	
Observe that the first statement of Theorem \ref{thm.Gen-insights} is equivalent to the forward direction of Theorem \ref{thm.vertices}, and the proof is also identical. The second statement of Theorem \ref{thm.Gen-insights} is equivalent to the reverse direction of Theorem \ref{thm.vertices} and again the proof is identical.

The proofs of the reverse direction of Theorem \ref{thm.vertices} and the forward direction of Theorem \ref{thm.edges} do not hold for general graphs. Counter-examples to these statements are given in Example \ref{ex.6} and Example \ref{ex.7}, respectively. In bipartite graphs, both these facts use Corollary \ref{cor.Birkhoff}, which follows from Theorem \ref{thm.Birkhoff}. The latter fact does not hold for LP (\ref{eq.core-primal-general}); however, a weaker fact, given   in Theorem \ref{thm.Balinski}, holds.

\begin{example}
	\label{ex.6}
Consider the game depicted in Figure \ref{fig.6}. Let the weights of edges $(v_1, v_2)$, $(v_2, v_3)$, $(v_3, v_1)$ and $(v_1, v_4)$ be $1.5, 1, 1.5$ and $1$, respectively. Then an optimal integral matching is $\{ (v_1, v_4), (v_2, v_3) \}$, having weight 2. The three-cycle $v_1, v_2, v_3$, taken to the extent of half, is a fractional matching of the same weight. Therefore, this graph has non-empty core. It has a unique core imputation which assigns profits of $1, {1 \over 2}, {1 \over 2}, 0$ to $v_1, v_2, v_3, v_4$, respectively. Since $v_4$ is essential but is not paid in the unique core imputation, this game provides a counter-example to the reverse direction of Theorem \ref{thm.vertices} in general graphs.  
\end{example}


\begin{figure}[h]
\begin{center}
\includegraphics[width=2.4in]{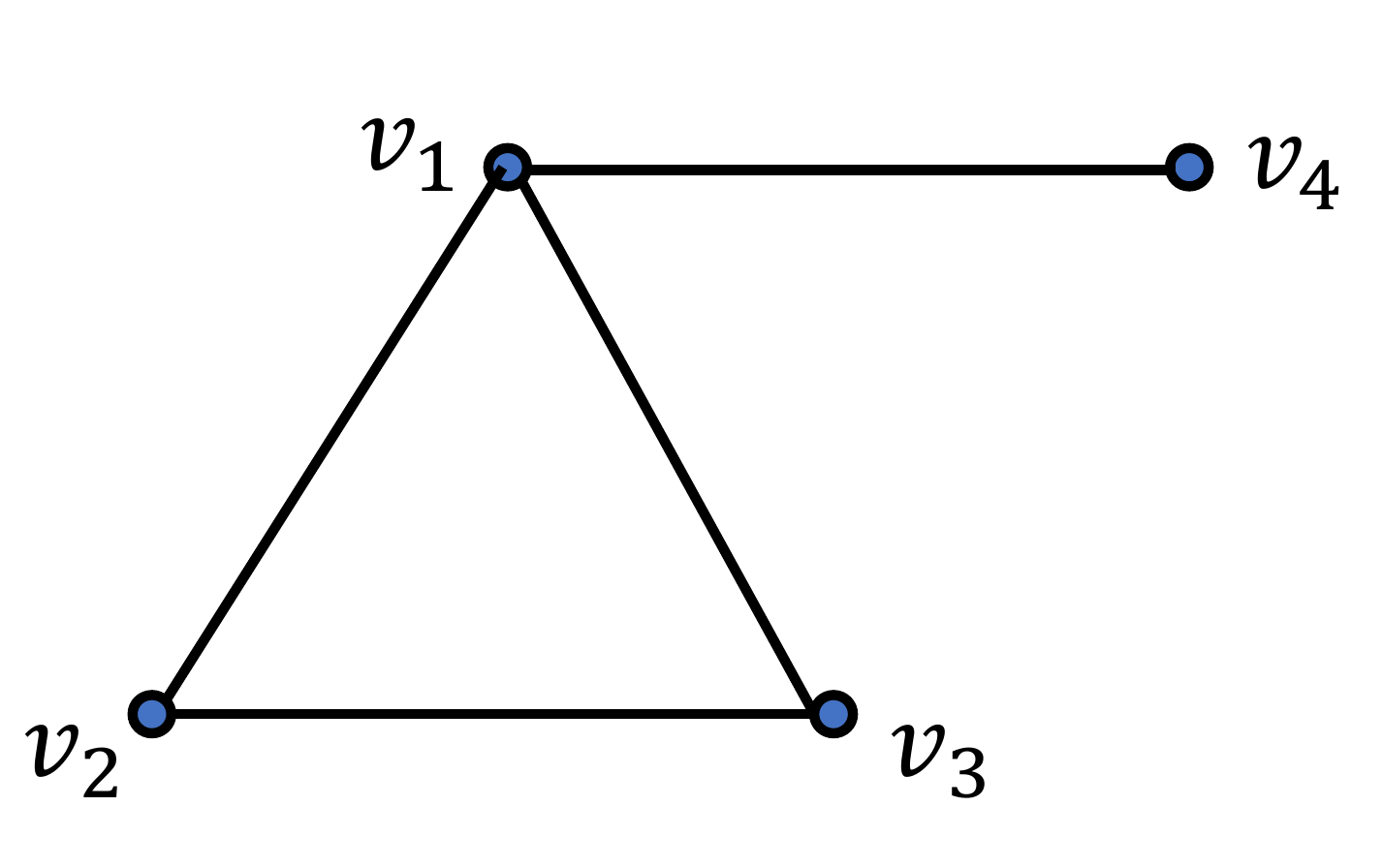}
\caption{The counter-example related to Theorem \ref{thm.Gen-insights}.}
\label{fig.6}
\end{center}
\end{figure}


Since the first statement of Theorem \ref{thm.Gen-insights} is weaker than Theorem \ref{thm.vertices}, the following corollary is weaker than Corollary \ref{cor.vertices}.  

\begin{corollary}
	\label{cor.Gen-vertices}
	In a concurrent game, the set of essential players is non-empty and in every core imputation, the entire worth of the game is distributed among essential players. 
\end{corollary} 

The contrapositive of the second statement of Theorem \ref{thm.Gen-insights} is the following:

 \begin{center}
{\em If a team is sometimes overpaid, then it is subpar.}
\end{center}

\begin{example}
	\label{ex.7}
	In the game defined in Example \ref{ex.5}, the teams $(v_4, v_7), (v_1, v_2), (v_2, v_3)$ and $(v_1, v_7)$ are all subpar, since they are not matched in any maximum weight matching. Of these, only the first team is overpaid in the unique imputation in the core, the rest are fairly paid. On the other hand,	the vertices which get positive profit are precisely the essential vertices. 
\end{example}
 
Corollary \ref{cor.Gen-degen} is identical to Corollary \ref{cor.degen}.

\begin{corollary}
	\label{cor.Gen-degen}
	In the presence of degeneracy, imputations in the core of a concurrent game treat:
	\begin{itemize}
			\item  viable players in the same way as subpar players, namely they are never paid.  
		\item viable teams in the same way as essential teams, namely they are always fairly paid. 
	\end{itemize}
\end{corollary}

%% file: b-Matching.tex
\section{The Core of Bipartite $b$-Matching Games}
\label{sec.b-matching-game}

In this section, we will define two versions of the {\em bipartite $b$-matching game} and we will study their core imputations; both versions generalize the assignment game.

\subsection{Definitions and Preliminary Facts}
\label{sec.b-prelim}

As in the assignment game, let $G = (U, V, E), \ w: E \rightarrow \cR_+$ be the underlying bipartite graph and edge-weight function. Let function $b: U \cup V \rightarrow \ZZ_+$ give a bound on the number of times a vertex can be matched. Under the {\em unconstrained bipartite $b$-matching game}, each edge can be matched multiple number of times and under the {\em constrained bipartite $b$-matching game}, each edge can be matched at most once. Observe that even in the first version, limits imposed by $b$ on vertices will impose limits on edges  --- thus edge $(i, j)$ can be matched at most $\min \{b_i, b_j\}$ times. 

The {\em worth} of a coalition $(S_u \cup S_v)$, with $S_u \subseteq U, S_v \subseteq V$, is the weight of a maximum weight $b$-matching in the graph $G$ restricted to vertices in $(S_u \cup S_v)$ only; we will denote this by $p(S_u \cup S_v)$. Whether an edge can be matched at most once or more than once depends on the version of the problem we are dealing with. $p(U \cup V)$ is called the {\em worth of the game}. The {\em characteristic function} of the game is defined to be $p: 2^{U \cup V} \rightarrow \cR_+$. Definitions \ref{def.imputation} and \ref{def.core}, defining an imputation and the core, carry over unchanged from the assignment game.  

The tennis setting, given in the Introduction, provides a vivid description of these two variants  of the $b$-matching game as well. Let $K$ denote the maximum $b$-value of a vertex and assume that the tennis club needs to enter mixed doubles teams into $K$ tennis tournaments. In the first variant, a team can play in multiple tournaments and in the second version, a team can play in at most one tournament. In both cases, a player $i$ can play in at most $b_i$ tournaments. The goal of the tennis club is to maximize its profit over all the tournaments and hence picks a maximum weight $b$-matching in $G$. An imputation in the core gives a way of distributing the profit in such a way that no sub-coalition has an incentive to secede. 

Linear program (\ref{eq.b-uncon-core-primal-bipartite}) gives the LP-relaxation of the problem of finding a maximum weight $b$-matching for the unconstrained version. In this program, variable $x_{ij}$ indicates the extent to which edge $(i, j)$ is picked in the solution; observe that  there is no upper bound on the variables $x_{ij}$ since an edge can be matched any number of times.

	\begin{maxi}
		{} {\sum_{(i, j) \in E}  {w_{ij} x_{ij}}}
			{\label{eq.b-uncon-core-primal-bipartite}}
		{}
		\addConstraint{\sum_{(i, j) \in E} {x_{ij}}}{\leq b_i \quad}{\forall i \in U}
		\addConstraint{\sum_{(i, j) \in E} {x_{ij}}}{\leq b_j }{\forall j \in V}
		\addConstraint{x_{ij}}{\geq 0}{\forall (i, j) \in E}
	\end{maxi}

Taking $u_i$ and $v_j$ to be the dual variables for the first and second constraints of (\ref{eq.b-uncon-core-primal-bipartite}), we obtain the dual LP: 

 	\begin{mini}
		{} {\sum_{i \in U}  {b_i u_{i}} + \sum_{j \in V} {b_j v_j}} 
			{\label{eq.b-uncon-core-dual-bipartite}}
		{}
		\addConstraint{ u_i + v_j}{ \geq w_{ij} \quad }{\forall (i, j) \in E}
		\addConstraint{u_{i}}{\geq 0}{\forall i \in U}
		\addConstraint{v_{j}}{\geq 0}{\forall j \in V}
	\end{mini}

	\bigskip

Linear program  (\ref{eq.b-con-core-primal-bipartite}) gives the LP-relaxation of the problem of finding a maximum weight $b$-matching for the constrained version. Observe that in this program, variables $x_{ij}$ are upper bounded by 1, since an edge can be matched at most once.

		\begin{maxi}
		{} {\sum_{(i, j) \in E}  {w_{ij} x_{ij}}}
			{\label{eq.b-con-core-primal-bipartite}}
		{}
		\addConstraint{\sum_{(i, j) \in E} {x_{ij}}}{\leq b_i \quad}{\forall i \in U}
		\addConstraint{\sum_{(i, j) \in E} {x_{ij}}}{\leq b_j }{\forall j \in V}
		\addConstraint{x_{ij}}{\leq 1}{\forall (i, j) \in E}
		\addConstraint{x_{ij}}{\geq 0}{\forall (i, j) \in E}
	\end{maxi}

\begin{remark}
\label{rem.b-unimodular}
In both in LPs (\ref{eq.b-uncon-core-primal-bipartite}) and (\ref{eq.b-con-core-primal-bipartite}),  the matrices of coefficients of the constraints are totally unimodular \cite{LP.book}, and therefore both LPs always have integral optimal solutions.
\end{remark}

Taking $u_i$, $v_j$ and $z_{ij}$ to be the dual variables for the first, second and third  constraints of (\ref{eq.b-con-core-primal-bipartite}), we obtain the dual LP: 

 	\begin{mini}
		{} {\sum_{i \in U}  {b_i u_{i}} + \sum_{j \in V} {b_j v_j} + \sum_{(i, j) \in E}  {z_{ij}}} 
			{\label{eq.b-con-core-dual-bipartite}}
		{}
		\addConstraint{u_i + v_j + z_{ij}}{ \geq w_{ij} \quad }{\forall (i, j) \in E}
		\addConstraint{u_{i}}{\geq 0}{\forall i \in U}
		\addConstraint{v_{j}}{\geq 0}{\forall j \in V}
		\addConstraint{z_{ij}}{\geq 0}{\forall (i, j) \in E}
	\end{mini}

\subsubsection{The Framework of Deng et al. \cite{Deng1999algorithms}}
\label{sec.Deng}

In this section, we present the framework of Deng et al. \cite{Deng1999algorithms}, which was mentioned in the Introduction, and point out why it does not apply to the two versions of the $b$-matching game. Let $T = \{1, \cdots, n\}$ be the set of $n$ agents of the game. Let $w \in \R^m_+$ be an $m$-dimensional non-negative real vector specifying the weights of certain objects; in the assignment game, the objects are edges of the underlying graph. Let $A$ be an $n \times m$ matrix with $0/1$ entries whose $i^{th}$ row corresponds to agent $i \in T$. Let $x$ be an $m$-dimensional vector of variables and $\mathbb{1}$ be the $n$-dimensional vector of all 1s. Assume that the worth of the game is given by the objective function value of following integer program. 

	\begin{maxi}
		{} {w \cdot x}
			{\label{eq.IP}}
		{}
		\addConstraint{}{Ax \leq \mathbb{1}}
		\addConstraint{}{x \in \{0, 1\}}
	\end{maxi}

Moreover, for a sub-coalition, $T' \subseteq T$ assume that its worth is given by the integer program obtained by replacing $A$ by $A'$ in (\ref{eq.IP}), where $A'$ picks the set of rows corresponding to agents in $T'$. The LP-relaxation of (\ref{eq.IP}) is:

	\begin{maxi}
		{} {w \cdot x}
			{\label{eq.Primal}}
		{}
		\addConstraint{}{Ax \leq \mathbb{1}}
		\addConstraint{}{x \geq 0}
	\end{maxi}

Deng et al. proved that if LP (\ref{eq.Primal}) always has an integral optimal solution, then the set of core imputations of this game is exactly the set of optimal solutions to the dual of LP (\ref{eq.Primal}). 

As stated in Remark \ref{rem.b-unimodular}, the matrices of coefficients of both LPs (\ref{eq.b-uncon-core-primal-bipartite}) and (\ref{eq.b-con-core-primal-bipartite}) are totally unimodular and therefore these LPs always have integral optimal solutions. However, they still don't fall in the above-stated framework because their right-hand-sides are $b$ values of the vertices and not $\mathbb{1}$. 


\input{Uniform}


\subsection{The Core of the Unconstrained Bipartite $b$-Matching Game}
\label{sec.b-uncon-core}

Let $I$ denote an instance of this game and let $C(I)$ denote its set of core imputations. We will  show in Theorem \ref{thm.b-uncon-SS} that corresponding to every optimal solution to the dual LP (\ref{eq.b-uncon-core-dual-bipartite}), there is an imputation in $C(I)$. Let $D(I)$ denote the set of all such core imputations. Since $D(I) \neq \emptyset$, we get Corollary \ref{cor.b-uncon} stating that the core of this game is non-empty. Next, we will give an instance $I$ such that $D(I) \subset C(I)$, i.e., unlike the assignment game, $I$ has core imputations that don't correspond to optimal solutions to the dual LP.

The correspondence between optimal solutions to the dual LP (\ref{eq.b-uncon-core-dual-bipartite}) and core imputations in $D(I)$ is as follows. Given an optimal solution $(u, v)$, define the profit allocation to $i \in U$ to be $\alpha_i = b_i \cdot u_i$ and that to $j \in V$ to be $\beta_j = b_j \cdot v_j$. 

\begin{theorem}
	\label{thm.b-uncon-SS}
	The profit-sharing method $(\alpha, \beta)$, which corresponds to an optimal solution $(u, v)$ to the dual LP (\ref{eq.b-uncon-core-dual-bipartite}), is an imputation in the core of the unconstrained bipartite $b$-matching game.
\end{theorem}

\begin{proof}
By Remark \ref{rem.b-unimodular}, LP (\ref{eq.b-uncon-core-primal-bipartite}) always has an optimal solution that is integral, i.e., there is always an optimal solution to this LP that is a maximum weight $b$-matching in $G$. Let $W$ be the weight of such a matching; clearly, $W = p(U \cup V)$.
	
	Since $(u, v)$ is an optimal solution to the dual LP, (\ref{eq.b-uncon-core-dual-bipartite}),  by the LP-Duality Theorem, 
	\[ \sum_{i \in U} {b_i u_i} + \sum_{j \in V} {b_j v_j} = \ W \ = \sum_{i \in U} {\alpha_i} + \sum_{j \in V} {\beta_j} .\]
	 Therefore the imputation $(\alpha, \beta)$ distributes the worth of the game among the agents. 
	 
Consider a sub-coalition $(S_u \cup S_v)$, with $S_u \subseteq U, S_v \subseteq V$. Let $G'$ denote the restriction of $G$ to the vertices in $(S_u \cup S_v)$ and let $E'$ be the edges of $G'$.  Let $x'$ denote a maximum weight unconstrained $b$-matching in $G'$. Corresponding to each edge $(i, j) \in E'$, $x'_{ij}$ is integral and the total profit which this sub-coalition can generate  by seceding is
\[ p(S_u \cup S_v) = \sum_{(i, j) \in E'} {w_{ij} x'_{ij}} \] 

We need to show that
\[ \sum_{(i, j) \in E'} {w_{ij} x'_{ij}} \leq \sum_{i \in S_u} {\alpha_i} + \sum_{j \in S_v} {\beta_j} ,\] 
thereby proving the theorem. 

Note that $x$ and $x'$ may differ on edges in $E'$. Clearly, $x'$ is a feasible solution to the restriction of LP (\ref{eq.b-uncon-core-primal-bipartite}) to $G'$. Furthermore, the restriction of $(u, v)$ to vertices in $G'$ is a feasible solution to the restriction of LP (\ref{eq.b-uncon-core-dual-bipartite}) to $G'$. In the proof given below, the first inequality follows from the constraint of the dual LP (\ref{eq.b-uncon-core-dual-bipartite}), and the second from the two constraints of the primal LP (\ref{eq.b-uncon-core-primal-bipartite}).

\[ \sum_{(i, j) \in E'} {w_{ij} x'_{ij}} \leq \sum_{(i, j) \in E'} {(u_i + v_j) \cdot x'_{ij}} \]
\[ = \sum_{i \in S_u} {\left( u_i \cdot \sum_{(i, j) \in E'} {x'_{ij}} \right) }  + \sum_{j \in S_v} {\left( v_j \cdot \sum_{(i, j) \in E'} {x'_{ij}} \right) } \] 
\[ \leq \sum_{i \in S_u} {b_i u_i} + \sum_{j \in S_v} {b_j v_j}  =  \sum_{i \in S_u} {\alpha_i} + \sum_{j \in S_v} {\beta_j} \] 
\end{proof}

\begin{corollary}
\label{cor.b-uncon}
	The core of the unconstrained bipartite $b$-matching game is always non-empty. 
\end{corollary}

\begin{remark}
	\label{rem.uncon-mapping}
	Observe that the mapping given from optimal solutions to the dual LP (\ref{eq.b-uncon-core-dual-bipartite}) to core imputations in $D(I)$ is a bijection. 
\end{remark}


\begin{figure}[h]
\begin{center}
\includegraphics[width=2.4in]{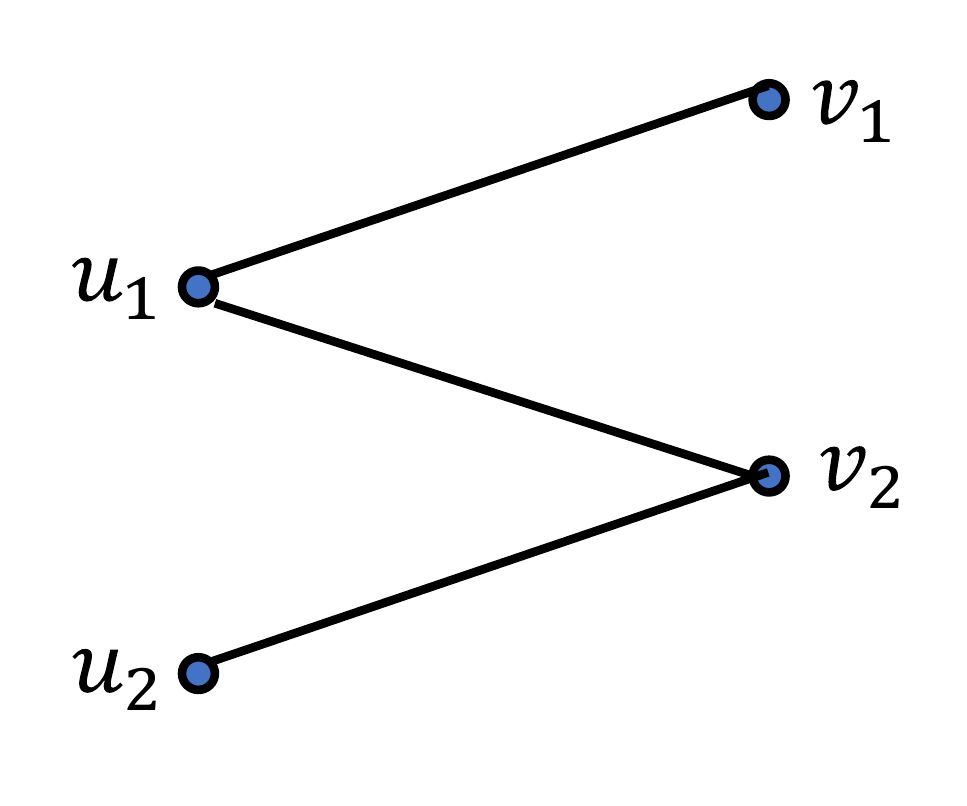}
\caption{The graph for Example \ref{ex.b-uncon}.}
\label{fig.7}
\end{center}
\end{figure}


\begin{example}
	\label{ex.b-uncon}
For the bipartite $b$-matching game defined by the graph of Figure \ref{fig.7}, let the $b$ values be $2, 1, 2, 1$ for $u_1, u_2, v_1, v_2$, and let the edge weights be $1, 3, 1$ for $(u_1, v_1), (u_1, v_2), (u_2, v_2)$. 
\end{example}

In this section, we will view the game defined in Example \ref{ex.b-uncon} as an unconstrained bipartite $b$-matching game and will show that it has a set of core imputations which do not correspond to optimal dual solutions, i.e., they lie in $C(I) - D(I)$. The optimal matching  picks edges $(u_1, v_1), (u_1, v_2)$ once each, for a total profit of 4. The unique optimal dual solution is $1, 0, 0, 2$ for $u_1, u_2, v_1, v_2$, and the corresponding core imputation is $2, 0, 0, 2$.

Let $\alpha_1, \alpha_2, \beta_1, \beta_2$ be the profits allocated to $u_1, u_2, v_1, v_2$. The solutions of the system of linear inequalities (\ref{eq.system}), for non-negative values of the variables, capture all possible core imputations, i.e., the set $C(I)$.  

\begin{equation}
\label{eq.system}
	\begin{aligned}
	\alpha_1 + \beta_1 & \geq 2\\
	\alpha_1 + \beta_2 & \geq 3\\
	\alpha_1 + \beta_1 + \beta_2 & \geq 4\\
	\alpha_2 + \beta_2 & \geq 1\\
	\alpha_1 + \alpha_2 + \beta_2 & \geq 3\\
	\alpha_1 + \alpha_2 + \beta_1 + \beta_2 & = 4
\end{aligned}
\end{equation}

On solving this system, we find that $\alpha_1, \alpha_2, \beta_1, \beta_2$ should be $1 + a, \ 0 \ b, \ 1+c$, where $a, b, c$ are non-negative and satisfy the system (\ref{eq.system2}). 

\begin{equation}
\label{eq.system2}
	\begin{aligned}
	a + b & \geq 1\\
	a + c & \geq 1\\
	a+ b + c & = 2
\end{aligned}
\end{equation}

A fourth constraint, $b \leq 1$ follows from the last two in this system. The solution $a = 1, \ b=0, \ c=1$ gives the core imputation corresponding to the unique optimal dual solution; the rest give the remaining core imputations, e.g., the imputation $3, 0, 0, 1$. 

For an arbitrary instance $I$, one can clearly capture all possible core imputations via an exponential sized system of inequalities  of the type $\geq$, one corresponding to each coalition $(S_u \cup S_v)$; its r.h.s. will be $p(S_u \cup S_v)$ and its l.h.s. will be the sum of all variables denoting profits accrued to vertices in this coalition. Note that all the  variables of this system will be constrained to be non-negative and it will have one equality corresponding to the worth of the grand coalition; the latter is the last equality in system  (\ref{eq.system}).

The following question arises: is there a smaller system which accomplishes this task? We observe that it suffices to include in the system only those coalitions whose induced subgraph is connected. This is so because if the induced subgraph for coalition $(S_u \cup S_v)$ has two or more connected components, then the sum of the inequalities for the connected components yields the inequality for coalition $(S_u \cup S_v)$. In particular, if the underlying graph of instance $I$ is sparse, this may lead to a much smaller system. Observe that the system (\ref{eq.system}), for Example \ref{ex.b-uncon}, follows from this idea. 

\begin{remark}
\label{rem.char}
	Since for the unconstrained bipartite $b$-matching game, the optimal dual solutions don't capture all core imputations, the characterizations established in Theorems \ref{thm.vertices} and \ref{thm.edges} for the assignment game, don't carry over. However, if one restricts to core imputations in the set $D(I)$ only, one can see that suitable modifications of these statements do hold. 
\end{remark}

%% file: Uniform.tex
\subsection{The Core of the Uniform Bipartite $b$-Matching Game}
\label{sec.b-uniform-core}

We first consider the special case of the bipartite $b$-matching game in which $b$ is the constant function is called  the {\em uniform bipartite $b$-matching game}; we will denote the constant by $b_c \in \ZZ_+$. The next theorem is analogous to the Shapley-Shubik Theorem.

\begin{theorem}
	\label{thm.b-uniform}
	For the uniform bipartite $b$-matching game, the dual completely characterizes its core.
\end{theorem}

\begin{proof}
The proof hinges on the fact that the polytope defined by the constraints of the primal LP, (\ref{eq.b-uncon-core-primal-bipartite}) has integral vertices, i.e., they are $b$-matchings in $G$. 

Let $(u, v)$ be an optimal dual solution. By integrality and the LP-duality theorem, the worth of the game, 
$$c(U \cup V) = b_c \cdot (\sum_{i \in U} {u_i} + \sum_{j \in V} {v_j}) .$$
Therefore $(\alpha, \beta)$ is an imputation, where $\alpha_i = b_c \cdot u_i$ and $\beta_j = b_c \cdot v_j$. 
	 
Consider a sub-coalition $(S_u \cup S_v)$, with $S_u \subseteq U, S_v \subseteq V$. Let $G'$ denote the subgraph of $G$ induced on the vertices $(S_u \cup S_v)$. Once again by integrality and the LP-duality theorem, we get that worth of $(S_u \cup S_v)$ equals the objective function value of  the optimal dual for graph  $G'$. Since the restriction of $(u, v)$ to $G'$ is a feasible dual for $G'$, we get that $b_c \cdot (\sum_{i \in S_u} {u_i} + \sum_{j \in S_v} {v_j}) =  \sum_{i \in S_u} {\alpha_i} + \sum_{j \in S_v} {\beta_j} \geq c(S_u \cup S_v)$, i.e., the core condition is satisfied for sub-coalition $(S_u \cup S_v)$. Therefore $(\alpha, \beta)$ is a core imputation.

Next, let $(\alpha, \beta)$ be a core imputation. By integrality and the LP-duality theorem,  
$$ \sum_{i \in U} {\alpha_i} + \sum_{j \in V} {\beta_j} \ = c(U \cup V) .$$
Let $u_i = {1 \over b_c} \alpha_i$ and $v_j = {1 \over b_c} \beta_j$. We will show that $(u, v)$ is an optimal dual solution for $G$, thereby proving the theorem.

Corresponding to any edge $e = (i, j)$, consider the sub-coalition $S = \{ i, j\}$. The worth of this sub-coalition is obtained by picking edge $e$ $b_c$ times, i.e., $c(\{i, j\}) = b_c \cdot w_e$.  Since $(\alpha, \beta)$ be a core imputation, the profit allocated to this sub-coalition is at least its worth, i.e., $\alpha_i + \beta_j \geq b_c \cdot w_e$. Dividing by $b_c$ we get $u_i + v_j \geq w_e$. Therefore $(u, v)$ satisfies the constraint in LP (\ref{eq.b-uncon-core-dual-bipartite}) and is hence a dual feasible solution. Since $(\alpha, \beta)$ is a core imputation, $ \sum_{i \in U} {\alpha_i} + \sum_{j \in V} {\beta_j} \ = c(U \cup V) .$ Therefore, by integrality and the LP-duality theorem, the objective function value of this dual equals the  optimal primal. Therefore $(u, v)$ is an optimal dual solution.  
\end{proof}

\begin{remark}
	\label{rem.simpler}
	Clearly, the proof given above can be used for characterizing the core of the assignment game as well. We note that it is simpler and more modular than the proof give in \cite{Shapley1971assignment}. 
\end{remark}

Next, we prove that the core of the uniform bipartite $b$-matching game also has two extreme imputations, as claimed for the assignment game in Theorem \ref{thm.extreme}. For $i \in U$, let $\alpha_i^h$ and $\alpha_i^l$ denote the highest and lowest profits that $i$ accrues among all imputations in the core. Similarly, for $j \in V$, let $\beta_j^h$ and $\beta_j^l$ denote the highest and lowest profits that $j$ accrues in the core. Let $\alpha^h$ and $\alpha^l$ denote the vectors whose components are $\alpha_i^h$ and $\alpha_i^l$, respectively. Similarly, let $\beta^h$ and $\beta^l$ denote vectors whose components are $\beta_j^h$ and $\beta_j^l$, respectively. 

\begin{theorem}
	\label{thm.b-extreme}	
	The core of the uniform bipartite $b$-matching game has two extreme imputations; they are $(\alpha^h, \beta^l)$ and $(\alpha^l, \beta^h)$.
\end{theorem} 

The proof of this theorem follows from Lemma \ref{lem.extreme-lemma}, whose proof is straightforward and is omitted. Let $(q, r)$ and $(s, t)$ be two imputations in the core of the uniform bipartite $b$-matching game. For each $i \in U$, let 
\[ \un{\alpha_i} = \min(q_i, s_i) \ \ \ \  \mbox{and} \ \ \ \  \ov{\alpha_i} = \max (q_i, s_i) .\]
Further, for each $j \in R$, let
\[  \un{\beta_j} = \min (r_j , t_j) \ \ \ \  \mbox{and} \ \ \ \  \ov{\beta_j} = \max (r_j, t_j) .\] 

\begin{lemma}
	\label{lem.extreme-lemma}
	$(\un{\alpha} , \ov{\beta} )$ and $(\ov{\alpha}, \un{\beta} )$ are imputations in the core of the uniform bipartite $b$-matching game.
\end{lemma}

%% file: b-Constrained.tex
\section{The Core of the Constrained Bipartite $b$-Matching Game}
\label{sec.b-con-core}

Our results for this game are related to, though not identical with, those for the unconstrained version. In Theorem \ref{thm.b-con-SS}, we will show that corresponding to every optimal solution to the dual LP (\ref{eq.b-con-core-dual-bipartite}), there is a set of core imputations. This theorem yields Corollary \ref{cor.b-con} stating that the core of this game is also non-empty. Finally, we will give an instance which has core imputations that don't correspond to optimal solutions to the dual LP.

The corresponding to an optimal solution to the dual LP (\ref{eq.b-uncon-core-dual-bipartite}), $(u, v, z)$, we define a set of imputations as follows. For each edge $(i, j)$ define two new  variables $c_{ij}$ and $d_{ij}$; both are constrained to be non-negative. Furthermore, consider all possible ways of splitting $z_{ij}$ into $c_{ij}$ and $d_{ij}$, i.e., $z_{ij} = c_{ij} + d_{ij}$. Define the profit allocation to $i \in U$ to be 
$$\alpha_i = b_i \cdot u_i + \sum_{j: (i, j) \in E} {c_{ij}}$$ 
and that to $j \in V$ to be 
$$\beta_j = b_j \cdot v_j + \sum_{i: (i, j) \in E} {d_{ij}} .$$ 
Taken over all possible ways of splitting all $z_{ij}$s, this gives a set of imputations. 

\begin{theorem}
	\label{thm.b-con-SS}
	All profit-sharing methods $(\alpha, \beta)$, which correspond to the optimal solution $(u, v, z)$ to the dual LP (\ref{eq.b-con-core-dual-bipartite}), are imputations in the core of the constrained bipartite $b$-matching game.
\end{theorem}

\begin{proof}
The proof is similar to that of Theorem \ref{thm.b-uncon-SS}, though it is more involved because of the additional variables. Again, by Remark \ref{rem.b-unimodular}, there is always an integral  optimal solution to LP (\ref{eq.b-con-core-primal-bipartite}), i.e., a solution that is a maximum weight $b$-matching in $G$. Let $W$ be the weight of such a matching; clearly, $W = p(U \cup V)$.
	
	For the rest of this proof, let $(\alpha, \beta)$ be one of the profit-sharing methods that corresponds to $(u, v, z)$; the latter being an optimal solution to the dual LP, (\ref{eq.b-uncon-core-dual-bipartite}). By the LP-Duality Theorem, and since each $z_{ij}$ is split among $\alpha_i$ and $\beta_j$, we get
	\[ \sum_{i \in U} {b_i u_i} + \sum_{j \in V} {b_j v_j} + \sum_{(i, j) \in E}  {z_{ij}} = \ W \ = \sum_{i \in U} {\alpha_i} + \sum_{j \in V} {\beta_j} .\]
	 Therefore the imputation $(\alpha, \beta)$ distributes the worth of the game, $W$, among the agents. 
	 
As before, consider a sub-coalition $(S_u \cup S_v)$, with $S_u \subseteq U, S_v \subseteq V$. Let $G'$ denote the restriction of $G$ to the vertices in $(S_u \cup S_v)$ and let $E'$ be the edges of $G'$. Let $x'$ denote a maximum weight constrained $b$-matching in $G'$. Corresponding to each edge $(i, j) \in E'$, $x'_{ij}$ is $0/1$ and the total profit which this sub-coalition can generate by seceding is
\[ p(S_u \cup S_v) = \sum_{(i, j) \in E'} {w_{ij} x'_{ij}} \] 

We need to show that
\[ \sum_{(i, j) \in E'} {w_{ij} x'_{ij}} \leq \sum_{i \in S_u} {\alpha_i} + \sum_{j \in S_v} {\beta_j} ,\] 
thereby proving the theorem. 

First, we simply matters by observing that 
since $x'$ is a 0/1 vector,  

\[ z_{ij} \cdot x'_{ij}  \leq z_{ij}  = c_{ij} + d_{ij} .\]

In the proof given below, the first inequality follows from the constraint of the dual LP (\ref{eq.b-con-core-dual-bipartite}), the second uses the above-stated fact that $z_{ij} \cdot x'_{ij} \leq c_{ij} + d_{ij}$, and the third follows from the first two constraints of the primal LP (\ref{eq.b-con-core-primal-bipartite}).

\[ \sum_{(i, j) \in E'} {w_{ij} x'_{ij}} \leq \sum_{(i, j) \in E'} {(u_i + v_j + z_{ij}) \cdot x'_{ij}} \]

\[ \leq \sum_{i \in S_u} { u_i  \cdot \left( \sum_{j: (i, j) \in E'} {x'_{ij}} \right)}
+ \sum_{j \in S_v} { v_j  \cdot \left( \sum_{i: (i, j) \in E'} {x'_{ij}} \right)} +
\sum_{(i, j) \in E'} {(c_{ij} + d_{ij})} . \] 

\[ \leq \sum_{i \in S_u} {b_i \cdot u_i} + \sum_{j \in S_v} {b_j \cdot v_j} +  
\sum_{(i, j) \in E'} {(c_{ij} + d_{ij})} . \]

\[ = \sum_{i \in S_u} {\left( b_i \cdot u_i +  \sum_{j: (i, j) \in E'} {c_{ij}} \right) } + \sum_{j \in S_v} { \left( b_j \cdot v_j +  \sum_{i: (i, j) \in E'} {d_{ij}} \right) }  =  \sum_{i \in S_u} {\alpha_i} + \sum_{j \in S_v} {\beta_j} \] 
\end{proof}

\begin{corollary}
\label{cor.b-con}
	The core of the constrained bipartite $b$-matching game is always non-empty. 
\end{corollary}

In this section, we will view the game defined in Example \ref{ex.b-uncon} as a constrained bipartite $b$-matching game and will again show that it has a set of core imputations which do not correspond to optimal dual solutions. The optimal matching picks edges $(u_1, v_1), (u_1, v_2)$ once each, for a total profit of 4. Unlike the unconstrained case, this time, the optimal dual is not unique. The optimal dual solutions are given by $1, 0, 0, 2-a$, for vertices $u_1, u_2, v_1, v_2$, and $0, a, 0$ for edges $(u_1, v_1), (u_1, v_2), (u_2, v_2)$, where $a \in [0, 1]$. The corresponding core imputations are $3 - b, 0, 0, 1+b$, for the four vertices $u_1, u_2, v_1, v_2$, where $b \in [0, 1]$.

As in the unconstrained case, let $\alpha_1, \alpha_2, \beta_1, \beta_2$ be the profits allocated to $u_1, u_2, v_1, v_2$. This time, the system of linear inequalities whose solutions capture all possible core imputations is given by system  (\ref{eq.system}) after replacing the first inequality by 
$$ \alpha_1 + \beta_1 \geq 1. $$
This is so because edge $(u_1, v_1)$ can be matched twice under the the unconstrained bipartite $b$-matching game, but only once under the constrained version. As before, non-negativity is imposed on all these variables. On solving this system, we find that $\alpha_1, \alpha_2, \beta_1, \beta_2$ should be $1, \ 0 \ b, \ 1+c$, where $a, b, c$ are non-negative and satisfy the system (\ref{eq.system3}).

\begin{equation}
\label{eq.system3}
	\begin{aligned}
	a + b & \geq 1\\
	a + c & \geq 2\\
	a+ b + c & = 3
\end{aligned}
\end{equation}

Solutions of this system which do not correspond to dual solutions include $1, 0, 0, 3$ and $0, 0, 1, 3$. Observe that neither of these is a core imputation for the unconstrained bipartite $b$-matching game. The method given in Section \ref{sec.b-uncon-core}, for finding a smaller system, holds for this case as well and so does Remark \ref{rem.char}.

\begin{remark}
	\label{rem.not-direct}
	In the assignment game, core imputations were precisely optimal dual solutions. On the other hand, in both versions of the bipartite $b$-matching game, core imputations are obtained from optimal dual solutions via specific operations. As stated in Remark \ref{rem.uncon-mapping}, for the unconstrained version, there is a bijection between optimal dual solutions and  core imputations in $D(I)$. In contrast, for the constrained version, the set of imputations corresponding to optimal dual solutions may not be disjoint. 
\end{remark}

Let us illustrate the last point of Remark \ref{rem.not-direct} via Example \ref{ex.b-uncon}. Consider the two optimal dual solutions obtained by setting $a = 0$ and $a = 1$, namely $1, 0, 0, 2$, for vertices $u_1, u_2, v_1, v_2$, and $0, 0, 0$ for edges $(u_1, v_1), (u_1, v_2), (u_2, v_2)$; and $1, 0, 0, 1$, for vertices $u_1, u_2, v_1, v_2$, and $0, 1, 0$ for edges $(u_1, v_1), (u_1, v_2), (u_2, v_2)$. Both these optimal duals yield the core imputation assigning profits of $2, 0, 0, 2$ for $u_1, u_2, v_1, v_2$.

%% file: b-General.tex
\section{The Core of a General Bipartite $b$-Matching Game}

\label{sec.b-gen-core}

We will use the following theorem of Hoffman and Kruskal \cite{Hoffman2010integral} to define a general bipartite $b$-matching game and show that every optimal dual solution leads to an imputation in its core. The only place we will sacrifice generality is in assuming that $c, d \in \ZZ_+^m$, since negative entries in these vectors are not very meaningful for our game. 

\begin{theorem}
\label{thm.Hoffman}
[Hoffman and Kruskal \cite{Hoffman2010integral}] 
Let $A$ be an $n \times m$ totally unimodular matrix. Then for all integral vectors $a, b \in \ZZ_+^n$ and $c, d \in \ZZ^m$, the polyhedron
\[ x \in \R^m \ \ s.t. \ \ a \leq Ax \leq b, \ \ c \leq x \leq d \]
has all integral vertices. Conversely, if this polyhedron has all integral vertices for every choice of integral vectors $a, b, c, d$ then matrix $A$ is totally unimodular.
\end{theorem}

As in the assignment game, let $G = (U, V, E), \ w: E \rightarrow \cR_+$ be the underlying bipartite graph and edge-weight function. Let functions $a: U \cup V \rightarrow \ZZ_+$ and $b: U \cup V \rightarrow \ZZ_+$ give a lower bound and an upper bound, respectively, on the number of times a vertex can be matched. Further, let functions $c: E \rightarrow \ZZ_+$ and $d: E\rightarrow \ZZ_+$ give a lower bound and an upper bound, respectively, on the number of times an edge can be matched.

The {\em worth} of a coalition $(S_u \cup S_v)$, with $S_u \subseteq U, S_v \subseteq V$, is the weight of a maximum weight $b$-matching in the graph $G$ restricted to vertices in $(S_u \cup S_v)$ only; we will denote this by $p(S_u \cup S_v)$. The constraints on edges and vertices of this subgraph are dictated from those in $G$. $p(U \cup V)$ is called the {\em worth of the game}. The {\em characteristic function} of the game is defined to be $p: 2^{U \cup V} \rightarrow \cR_+$. Definitions \ref{def.imputation} and \ref{def.core}, defining an imputation and the core, carry over unchanged from the assignment game.  

The tennis setting, given in the Introduction and used in Section \ref{sec.b-prelim}, is again useful for this setting. Let $K$ denote the maximum $b$-value of a vertex and assume that the tennis club needs to enter mixed doubles teams into $K$ tennis tournaments with the following constraints. A team $(i, j)$ needs to play in at least $c_{ij}$ and at most $d_{ij}$ tournaments. A player $i$ needs to play in at least $a_i$ and at most $b_i$ tournaments. Once again, the goal of the tennis club is to maximize its profit over all the tournaments and hence picks a maximum weight $b$-matching in $G$. An imputation in the core gives a way of distributing the profit in such a way that no sub-coalition has an incentive to secede.

Linear program  (\ref{eq.b-con-core-primal-bipartite}) gives the LP-relaxation of the problem of finding a maximum weight $b$-matching for the general version.

		\begin{maxi}
		{} {\sum_{(i, j) \in E}  {w_{ij} x_{ij}}}
			{\label{eq.b-gen-core-primal}}
		{}
		\addConstraint{\sum_{(i, j) \in E} {x_{ij}}}{\geq a_i \quad}{\forall i \in U}
		\addConstraint{\sum_{(i, j) \in E} {x_{ij}}}{\leq b_i \quad}{\forall i \in U}
		\addConstraint{\sum_{(i, j) \in E} {x_{ij}}}{\geq a_j \quad}{\forall j \in V}
		\addConstraint{\sum_{(i, j) \in E} {x_{ij}}}{\leq b_j \quad}{\forall j \in V}
		\addConstraint{x_{ij}}{\geq c_{ij}}{\forall (i, j) \in E}
		\addConstraint{x_{ij}}{\leq d_{ij}}{\forall (i, j) \in E}
	\end{maxi}

Taking $\alpha_i$, $\beta_i$, $\alpha_j$, $\beta_j$, $\gamma_{ij}$ and $\delta_{ij}$  to be the dual variables for the first to the sixth constraints, respectively, of (\ref{eq.b-gen-core-primal}), we obtain the dual LP: 

 	\begin{mini}
		{} {\sum_{i \in U}  {(b_i \beta_i - a_i \alpha_i)} + \sum_{j \in V} {(b_j \beta_j - a_j \alpha_j)} + \sum_{(i, j) \in E}  {(d_{ij} \delta_{ij} - c_{ij} \gamma_{ij})}} 
			{\label{eq.b-gen-core-dual}}
		{}
		\addConstraint{(\beta_i - \alpha_i) + (\beta_j - \alpha_j) + (\delta_{ij} - \gamma_{ij})}{\geq w_{ij} \quad }{\forall (i, j) \in E}
		\addConstraint{\alpha_{i},\beta_i}{\geq 0}{\forall i \in U}
		\addConstraint{\alpha_{j},\beta_j}{\geq 0}{\forall j \in V}
		\addConstraint{\gamma_{ij},\delta_{ij}}{\geq 0}{\forall (i, j) \in E}
	\end{mini}

Our results for this game are related to those for the constrained version. In Theorem \ref{thm.b-gen-SS}, we will show that corresponding to every optimal solution to the dual LP (\ref{eq.b-gen-core-dual}), there is a set of core imputations, thereby showing that the core of this game is also non-empty. Since we have given instances for the unconstrained and constrained versions which have core imputations that don't correspond to optimal solutions to the dual LP, the same holds for this game as well. 

The corresponding to an optimal solution to the dual LP (\ref{eq.b-gen-core-dual}), $(\alpha, \beta, \gamma, \delta)$, we define a set of imputations as follows. For each edge $(i, j)$, define four new variables $\gamma^i_{ij}, \gamma^j_{ij}, \delta^i_{ij}, \delta^j_{ij}$, all constrained to be non-negative. Consider all possible ways of splitting $\gamma_{ij}$ into $\gamma^i_{ij}, \gamma^j_{ij}$ and of splitting $\delta_{ij}$ into $\delta^i_{ij}, \delta^j_{ij}$.

Define the profit allocation to $i \in U$ to be 
$$\mu_i = {(b_i \beta_i - a_i \alpha_i)} + \sum_{j: (i, j) \in E} {(d_{ij}\delta^i_{ij} - c_{ij} \gamma^i_{ij})} $$   
and that to $j \in V$ to be 
$$\nu_j = {(b_j \beta_j - a_j \alpha_j)} + \sum_{i: (i, j) \in E} {(d_{ij}\delta^j_{ij} - c_{ij} \gamma^j_{ij})} .$$   
Taken over all possible ways of splitting all $\gamma_{ij}$s and $\delta_{ij}$s, this gives a set of imputations.

\begin{theorem}
	\label{thm.b-gen-SS}
	All profit-sharing methods $(\mu, \nu)$, which correspond to the optimal solution $(\alpha, \beta, \gamma, \delta)$ to the dual LP (\ref{eq.b-gen-core-dual}), are imputations in the core of the general bipartite $b$-matching game.
\end{theorem}

\begin{proof}
The proof is similar to that of Theorem \ref{thm.b-con-SS}, though it is more involved because of the additional variables. Again, by Theorem \ref{thm.Hoffman}, there is always an integral  optimal solution to LP (\ref{eq.b-gen-core-primal}). This is a maximum weight generalized $b$-matching in $G$. Let $W$ be the weight of such a matching; clearly, $W = p(U \cup V)$.
	
Let $(\mu, \nu)$ be one of the profit-sharing methods which correspond to the optimal dual solution $(\alpha, \beta, \gamma, \delta)$. As in the proof of Theorem \ref{thm.b-con-SS}, by the LP-Duality Theorem we get that imputation $(\mu, \nu)$ distributes the worth of the game, $W$, among the agents. 
	 
Consider a sub-coalition $(S_u \cup S_v)$, with $S_u \subseteq U, S_v \subseteq V$. Let $G'$ denote the restriction of $G$ to the vertices in $(S_u \cup S_v)$ and let $E'$ be the edges of $G'$. Let $x'$ denote a maximum weight constrained $b$-matching in $G'$. Corresponding to each edge $(i, j) \in E'$, $x'_{ij}$ is $0/1$ and the total profit which this sub-coalition can generate by seceding is
\[ p(S_u \cup S_v) = \sum_{(i, j) \in E'} {w_{ij} x'_{ij}} \] 

We need to show that
\[ \sum_{(i, j) \in E'} {w_{ij} x'_{ij}} \leq \sum_{i \in S_u} {\mu_i} + \sum_{j \in S_v} {\nu_j} ,\] 
thereby proving the theorem. 

First observe that since $c_{ij} \leq x'_{ij} \leq d_{ij}$, we get that 
\[ (\delta_{ij} - \gamma_{ij}) \cdot x'_{ij}  \leq d_{ij} \delta_{ij} - c_{ij} \gamma_{ij} .\]
Also, for each $i \in S_U$, by the first two constraints of the primal LP (\ref{eq.b-gen-core-primal}), $a_i \leq \sum_{j: (i, j) \in E'} {x'_{ij}} \leq b_i$. Therefore we get that 
\[ {(\beta_i - \alpha_i) \cdot \left( \sum_{j: (i, j) \in E'} {x'_{ij}} \right)} \leq b_i \beta_{ij} - a_i \alpha_{ij} .\]
An analogous statement holds for each $j \in S_V$.

In the proof given below, the first inequality follows from the constraint of the dual LP (\ref{eq.b-gen-core-dual}), and the second and third follow from the facts stated above. 

\[ \sum_{(i, j) \in E'} {w_{ij} x'_{ij}} \leq \sum_{(i, j) \in E'} {((\beta_i - \alpha_i) + (\beta_j - \alpha_j) + (\delta_{ij} - \gamma_{ij})) \cdot x'_{ij}} \]

\[ \leq \sum_{i \in S_u} {(\beta_i - \alpha_i) \cdot \left( \sum_{j: (i, j) \in E'} {x'_{ij}} \right)}
+ \sum_{j \in S_v} {(\beta_j - \alpha_j) \cdot \left( \sum_{i: (i, j) \in E'} {x'_{ij}} \right)} +
\sum_{(i, j) \in E'} {(d_{ij} \delta_{ij} - c_{ij} \gamma_{ij})} . \] 

\[ \leq \sum_{i \in S_u} {(b_i \beta_i - a_i \alpha_i)} + \sum_{j \in S_v}  + {(b_j \beta_j - a_j \alpha_j)} + \sum_{(i, j) \in E'} {(d_{ij} \delta_{ij} - c_{ij} \gamma_{ij})} . \]

\[ = \sum_{i \in S_u} {\left( {(b_i \beta_i - a_i \alpha_i)} + \sum_{j: (i, j) \in E'} {(d_{ij}\delta^i_{ij} - c_{ij} \gamma^i_{ij})} \right) } + \sum_{j \in S_v} { \left( {(b_j \beta_j - a_j \alpha_j)} +  \sum_{i: (i, j) \in E'} {(d_{ij}\delta^j_{ij} - c_{ij} \gamma^j_{ij})} \right) } \]
\[ =  \sum_{i \in S_u} {\mu_i} + \sum_{j \in S_v} {\nu_j} \] 
\end{proof}

\begin{corollary}
\label{cor.b-gen}
	The core of the general bipartite $b$-matching game is always non-empty. 
\end{corollary}

As stated above, since we have given instances for the unconstrained and constrained bipartite $b$-matching games which have core imputations that don't correspond to optimal solutions to their dual LPs, the same holds for this game as well.

%% file: discussion.tex
\section{Discussion}
\label{sec.discussion}

Our most important open question is to shed light on the origins of core imputations, for the two bipartite $b$-matching games, which do not correspond to optimal dual solutions. Is there a ``mathematical structure'' that produces them? A related question is to determine the complexity of the following question for these two games: Given an imputation for a game, decide if it belongs to the core. We believe this question should be co-NP-complete. On the other hand, the following question is clearly in P: Given an imputation for a game $I$, decide if it lies in $D(I)$. 

As stated in Section \ref{sec.matching-game}, for the assignment game, Shapley and Shubik were able to characterize ``antipodal'' points in the core. An analogous understanding of the core of the general graph matching games having non-empty core will be desirable. 

For the assignment game, Demange, Gale and Sotomayor \cite{Demange1986multi} give an auction-based procedure to obtain a core imputation; it turns out to be optimal for the side that proposes, as was the case for the deferred acceptance algorithm of Gale and Shapley \cite{GaleS} for stable matching. Is there an analogous procedure for obtaining an imputation in the core of the general graph matching games having non-empty core?

%% file: ack.tex
\section{Acknowledgements}
\label{sec.ack}

I wish to thank Herv\'e Moulin for asking the interesting question of extending results obtained for the assignment game to general graph matching games having a non-empty core; these results are presented in Section \ref{sec.general}. I also wish to thank Federico Echenique, Herv\'e Moulin and Thorben Trobst for several valuable discussions.

%% file: ms.bbl
\begin{thebibliography}{FKFH98}

\bibitem[AM85]{Talmud-Aumann1985game}
Robert~J Aumann and Michael Maschler.
\newblock Game theoretic analysis of a bankruptcy problem from the talmud.
\newblock {\em Journal of economic theory}, 36(2):195--213, 1985.

\bibitem[Bal65]{Balinski1965integer}
Michel~Louis Balinski.
\newblock Integer programming: methods, uses, computations.
\newblock {\em Management science}, 12(3):253--313, 1965.

\bibitem[Bir46]{Birkhoff1946three}
Garrett Birkhoff.
\newblock Three observations on linear algebra.
\newblock {\em Univ. Nac. Tacuman, Rev. Ser. A}, 5:147--151, 1946.

\bibitem[BKP12]{Biro2012computing}
P{\'e}ter Bir{\'o}, Walter Kern, and Dani{\"e}l Paulusma.
\newblock Computing solutions for matching games.
\newblock {\em International journal of game theory}, 41(1):75--90, 2012.

\bibitem[BST05]{Branzei2005strongly}
Rodica Br{\^a}nzei, Tam{\'a}s Solymosi, and Stef Tijs.
\newblock Strongly essential coalitions and the nucleolus of peer group games.
\newblock {\em International Journal of Game Theory}, 33(3):447--460, 2005.

\bibitem[CE15]{Chambers2015core}
Christopher~P Chambers and Federico Echenique.
\newblock The core matchings of markets with transfers.
\newblock {\em American Economic Journal: Microeconomics}, 7(1):144--64, 2015.

\bibitem[DGS86]{Demange1986multi}
Gabrielle Demange, David Gale, and Marilda Sotomayor.
\newblock Multi-item auctions.
\newblock {\em Journal of political economy}, 94(4):863--872, 1986.

\bibitem[DIN99]{Deng1999algorithms}
Xiaotie Deng, Toshihide Ibaraki, and Hiroshi Nagamochi.
\newblock Algorithmic aspects of the core of combinatorial optimization games.
\newblock {\em Mathematics of Operations Research}, 24(3):751--766, 1999.

\bibitem[Edm65]{Edmonds.matching}
J.~Edmonds.
\newblock Maximum matching and a polyhedron with 0,1-vertices.
\newblock {\em Journal of Research of the National Bureau of Standards B},
  69B:125--130, 1965.

\bibitem[EK01]{Eriksson2001stable}
Kimmo Eriksson and Johan Karlander.
\newblock Stable outcomes of the roommate game with transferable utility.
\newblock {\em International Journal of Game Theory}, 29(4):555--569, 2001.

\bibitem[FKFH98]{Faigle1998nucleon}
Ulrich Faigle, Walter Kern, S{\'a}ndor~P Fekete, and Winfried Hochst{\"a}ttler.
\newblock The nucleon of cooperative games and an algorithm for matching games.
\newblock {\em Mathematical Programming}, 83(1):195--211, 1998.

\bibitem[GLS88]{GLS}
M.~Grotschel, L.~Lovasz, and A.~Schirjver.
\newblock {\em Geometric Algorithms and Combinatorial Optimization}.
\newblock Springer-Verlag, 1988.

\bibitem[GS62]{GaleS}
David Gale and Lloyd~S Shapley.
\newblock College admissions and the stability of marriage.
\newblock {\em The American Mathematical Monthly}, 69(1):9--15, 1962.

\bibitem[HK10]{Hoffman2010integral}
Alan~J Hoffman and Joseph~B Kruskal.
\newblock Integral boundary points of convex polyhedra.
\newblock In {\em 50 Years of integer programming 1958-2008}, pages 49--76.
  Springer, 2010.

\bibitem[KP03]{Kern2003matching}
Walter Kern and Dani{\"e}l Paulusma.
\newblock Matching games: the least core and the nucleolus.
\newblock {\em Mathematics of operations research}, 28(2):294--308, 2003.

\bibitem[KPT20]{Konemann2020computing}
Jochen K{\"o}nemann, Kanstantsin Pashkovich, and Justin Toth.
\newblock Computing the nucleolus of weighted cooperative matching games in
  polynomial time.
\newblock {\em Mathematical Programming}, 183(1):555--581, 2020.

\bibitem[KTZ21]{b-matching-Konemann}
Jochen K{\"o}nemann, Justin Toth, and Felix Zhou.
\newblock On the complexity of nucleolus computation for bipartite b-matching
  games.
\newblock In {\em International Symposium on Algorithmic Game Theory}, pages
  171--185. Springer, 2021.

\bibitem[LP86]{LP.book}
L.~Lov\'{a}sz and M.D. Plummer.
\newblock {\em Matching Theory}.
\newblock North-Holland, Amsterdam--New York, 1986.

\bibitem[Mou14]{Moulin2014cooperative}
Herv{\'e} Moulin.
\newblock {\em Cooperative microeconomics: a game-theoretic introduction},
  volume 313.
\newblock Princeton University Press, 2014.

\bibitem[MPS79]{Leastcore-Maschler1979geometric}
Michael Maschler, Bezalel Peleg, and Lloyd~S Shapley.
\newblock Geometric properties of the kernel, nucleolus, and related solution
  concepts.
\newblock {\em Mathematics of operations research}, 4(4):303--338, 1979.

\bibitem[NR08]{Nunez-Dimension}
Marina N{\'u}{\~n}ez and Carles Rafels.
\newblock On the dimension of the core of the assignment game.
\newblock {\em Games and Economic Behavior}, 64(1):290--302, 2008.

\bibitem[Sch69]{Schmeidler1969nucleolus}
David Schmeidler.
\newblock The nucleolus of a characteristic function game.
\newblock {\em SIAM Journal on applied mathematics}, 17(6):1163--1170, 1969.

\bibitem[Sch86]{Sch-book}
A.~Schrijver.
\newblock {\em Theory of Linear and Integer Programming}.
\newblock John Wiley \& Sons, New York, NY, 1986.

\bibitem[Sot92]{Sotomayor1992multiple}
Marilda Sotomayor.
\newblock The multiple partners game.
\newblock In {\em Equilibrium and dynamics}, pages 322--354. Springer, 1992.

\bibitem[SR94]{Nucleolus-Assign-Solymosi1994algorithm}
Tam{\'a}s Solymosi and Tirukkannamangai E.~S. Raghavan.
\newblock An algorithm for finding the nucleolus of assignment games.
\newblock {\em International Journal of Game Theory}, 23(2):119--143, 1994.

\bibitem[SS71]{Shapley1971assignment}
Lloyd~S Shapley and Martin Shubik.
\newblock The assignment game {I}: The core.
\newblock {\em International Journal of Game Theory}, 1(1):111--130, 1971.

\bibitem[Vaz01]{ApproximationAlgs}
Vijay~V Vazirani.
\newblock {\em Approximation algorithms}.
\newblock Springer, 2001.

\bibitem[Vaz22]{Va.general}
Vijay~V Vazirani.
\newblock The general graph matching game: Approximate core.
\newblock {\em Games and Economic Behavior}, 132, 2022.

\bibitem[XLF21]{b-matching-approximate}
Han Xiao, Tianhang Lu, and Qizhi Fang.
\newblock Approximate core allocations for multiple partners matching games.
\newblock {\em arXiv preprint arXiv:2107.01442}, 2021.

\end{thebibliography}
